\documentclass[aps, prd, twocolumn, superscriptaddress, nofootinbib, amsmath, amssymb]{revtex4-2}

\usepackage{graphicx}  
\usepackage{mathrsfs}
\usepackage{amsthm}
\usepackage{bbold} 
\usepackage{braket}
\usepackage{physics}
\usepackage[normalem]{ulem}
\usepackage[breaklinks=true,colorlinks,citecolor=blue,linkcolor=blue,urlcolor=blue]{hyperref}
\usepackage{subcaption}
\usepackage{hhline}
\usepackage{multirow}
\newcommand{\matone}{{\mathbb{1}}}
\newtheorem{theorem}{Theorem}

\usepackage{xargs} 

\usepackage[font=small, labelfont=bf, textfont={small,it}]{caption}
\DeclareMathOperator{\Inv}{Inv}
\DeclareMathOperator{\Ker}{Ker}
\DeclareMathOperator{\Ln}{Ln}

\makeatletter
\newcommand{\fmarki}{*}
\newcommand{\fmarkii}{\ensuremath{\dagger}}
\newcommand{\fmarkiii}{\ensuremath{\ddagger}}
\newcommand{\fmarkiv}{\ensuremath{\mathsection}}
\newcommand{\fmarkv}{\ensuremath{\mathparagraph}}
\newcommand{\fmarkvi}{\ensuremath{\|}}
\newcommand{\fmarkvii}{**}
\newcommand{\fmarkviii}{\ensuremath{\dagger\dagger}}
\newcommand{\fmarkix}{\ensuremath{\ddagger\ddagger}}
                
\def\@fnsymbol#1{{\ifcase#1\or \fmarki\or \fmarkii\or \fmarkiii\or \fmarkiv\or \fmarkv\or \fmarkvi\or \fmarkvii\or \fmarkviii\or \fmarkix \else\@ctrerr\fi}}
\makeatother

\renewcommand{\fmarkix}{\ensuremath\ddag\ddag}
\def\@fnsymbol#1{{\ifcase#1\or \fmarki\or \fmarkii\or \fmarkiii\or \fmarkiv\or \fmarkv\or \fmarkvi\or \fmarkvii\or \fmarkviii\or \fmarkix \else\@ctrerr\fi}}
\makeatother

\begin{document}

\title{Quantum Computation of Thermal Averages for a Non-Abelian $D_4$ Lattice Gauge Theory via Quantum Metropolis Sampling}

\author{Edoardo~Ballini}
\email{edoardo.ballini@unitn.it}
\affiliation{Pitaevskii BEC Center and Department of Physics, University of Trento, Via Sommarive 14, I-38123 Trento, Italy}
\affiliation{INFN-TIFPA, Trento Institute for Fundamental Physics and Applications, Trento, Italy}

\author{Giuseppe~Clemente}
\email{giuseppe.clemente@desy.de}
\affiliation{Deutsches Elektronen-Synchrotron (DESY), Platanenallee 6, 15738 Zeuthen, Germany}

\author{Massimo~D'Elia}
\email{massimo.delia@unipi.it}
\affiliation{Dipartimento di Fisica dell'Universit\`a di Pisa and INFN --- Sezione di Pisa, Largo Pontecorvo 3, I-56127 Pisa, Italy.}

\author{Lorenzo~Maio}
\email{lorenzo.maio@phd.unipi.it}
\affiliation{Dipartimento di Fisica dell'Universit\`a di Pisa and INFN --- Sezione di Pisa, Largo Pontecorvo 3, I-56127 Pisa, Italy.}

\author{Kevin~Zambello}
\email{kevin.zambello@pi.infn.it}
\affiliation{Dipartimento di Fisica dell'Universit\`a di Pisa and INFN --- Sezione di Pisa, Largo Pontecorvo 3, I-56127 Pisa, Italy.}

\begin{abstract}
In this paper, we show the application of the Quantum Metropolis Sampling (QMS) algorithm to a toy gauge theory 
with discrete non-Abelian gauge group $D_4$ in (2+1)-dimensions,
discussing in general how some components of hybrid quantum-classical algorithms 
should be adapted in the case of gauge theories. 
In particular, we discuss the construction of random unitary operators which preserve gauge invariance and act transitively on the physical Hilbert space, 
constituting an ergodic set of quantum Metropolis moves between gauge invariant eigenspaces,
and introduce a protocol for gauge invariant measurements. 
Furthermore, we show how a finite resolution in the energy measurements 
distorts the energy and plaquette distribution measured via QMS, 
and propose a heuristic model that takes into account part of the deviations 
between numerical results and exact analytical results,
whose discrepancy tends to vanish by increasing the number of qubits used for the energy measurements.
\end{abstract}

\maketitle

\section{Introduction}\label{sec:intro}
In recent decades, the application of Monte Carlo simulations on classical computers has 
proven to be a powerful approach in the investigation of properties of quantum field theories. 
Despite that, some regimes still appear not to be accessible efficiently, 
especially in cases where the standard path integral formulation, based on a quantum-to-classical mapping 
(Trotter-Suzuki decomposition~\cite{Trotter:1959,Suzuki:1976be}), results 
in an algorithmic sign problem.
Such a problem prevents, for example, a deeper understanding of the QCD phase diagram 
with a finite baryonic chemical potential 
term~\cite{Shapiro:1983,Rajagopal:2001,Philipsen:2010gj,Ding:2015,Aarts:2015tyj} 
or with a topological theta term~\cite{Gross:1995bp,Mannel:2007zz,Unsal:2012zj}.
Recent advancements in quantum computing hardware and software give hope that the sign problem 
can be avoided by directly using the quantum formulation of the theories under study, 
therefore without the need for a quantum-to-classical mapping.
In these regards, we consider the task of computing thermal averages of observables 
(i.e. Hermitian operators), 
which are essential for characterizing the phase diagram 
of lattice quantum field theory and condensed matter systems.
The thermal average for an observable $\hat{O}$ at inverse temperature $\beta$ is defined as
\begin{align}\label{eq:thermaver}
    {\expval{\hat{O}}}_\beta = \Tr \lbrack \hat{O}\hat{\rho}_\beta \rbrack, 
    \qquad \hat{\rho}_\beta 
    = \frac{e^{-\beta H}}{Z},
\end{align}
where $\hat{\rho}_\beta$ represents the density matrix of the system, defined in terms of its Hamiltonian $H$, while $Z=\Tr \lbrack e^{-\beta H} \rbrack$.
In the last two decades, different quantum algorithms have been proposed 
for the task of thermal average estimation 
or thermal state preparation~\cite{Lu_2021,Yamamoto:2022jes,Selisko_2022,Davoudi:2022uzo,Ball_2022,Fromm_2023,Poulin_2009,Bilgin_2010,Riera_2012,Verdon_2019,Wu_2019,Zhu_2020,Motta_2020,Sun_2021,Powers_2023}. 
In this work, we focus on studying a non-Abelian lattice gauge theory toy model: 
a finite gauge group $D_4$ in $2+1$ dimensions.
While the system we investigate is not generally affected by a sign problem, 
some formal complications arise from the need to ensure the gauge invariance 
for a Markov Chain Monte Carlo method. 
Indeed, being this a gauge theory, 
we are actually interested in the space $\mathcal{H}_{\text{phys}}$ of gauge invariant 
(also called \emph{physical}) states, representing only a subspace of the full extended space 
$\mathcal{H}_{\text{ext}}$ which is used to define the dynamical variables of the system. 
So, the actual physical density matrix we consider is
\begin{align} 
\hat{\rho}^{(\text{phys})}_\beta 
= \frac{e^{-\beta H}\rvert_{\mathcal{H}_{\text{phys}}}}{\Tr \lbrack e^{-\beta H}\rvert_{\mathcal{H}_{\text{phys}}} \rbrack},
\end{align}
while the gauge constraints have to be encoded in the algorithm.
In this paper, we consider the Quantum Metropolis Sampling (QMS) algorithm~\cite{QMS_paper} 
to compute thermal averages of the system.
However, unlike what happens with systems that allow an unconstrained thermal estimation as expressed by Eq.~\eqref{eq:thermaver} 
(see Refs.~\cite{Qubipf_QMS,Aiudi:2023cyq} for the application of QMS to these cases),
here we focus on the new challenges emerging from requiring the gauge invariance constraint at each step, 
such as how to select a set of gauge invariant and ergodic Metropolis moves 
and how to perform gauge invariant measurements.
Our analysis takes into account the systematic errors of the algorithm,
but it does not include sources of error induced by quantum noise.
In particular, since our results have been produced using a noiseless emulator,
gauge invariance can be exactly preserved, 
so we do not need to consider how it would be broken by quantum noise 
(see Refs.~\cite{Stryker:2018efp,Raychowdhury:2018osk,Halimeh:2020ecg,Lamm:2020jwv,VanDamme:2021njp,Mathew:2022nep,Halimeh:2022mct,Gustafson:2023swx} 
for discussions about gauge-symmetry protection in noisy frameworks).

Sec.~\ref{sec:sys} introduces the system under investigation,
while a supplementary review of different (but compatible) formulations possible for a lattice gauge theory with a general finite gauge group is presented in Appendix~\ref{app:transfmat_gropth}.
In Sec.~\ref{sec:algo} we give a brief overview of the QMS algorithm 
and how it has to be adapted in general in order to preserve gauge invariance at each step, for both evolution (Sec.~\ref{sec:sys}), Metropolis updates (Sec.~\ref{subsec:GImoves}), and measurements (Sec.~\ref{subsec:retherm-GImeas}). 
Appendix~\ref{app:twogens} contains the statement and sketch of proof of a theorem used in Sec.~\ref{subsec:QMS} to build a set of gauge invariant Metropolis updates introduced in Sec.~\ref{subsec:GImoves} and guarantee its ergodicity.
Numerical results for the thermal energy distribution and plaquette measurement 
are displayed in Sec.~\ref{sec:numres}. In order to assess and visualize the accuracy of the measured thermal energy distributions for different numbers of qubits for the energy resolution,
we use Kernel Density Estimators, which are briefly reviewed in Appendix~\ref{app:GKDE}.
Finally, conclusions and future perspectives are discussed in Sec.~\ref{sec:conclusions}.

\section{The system}\label{sec:sys}

In this Section, we introduce the system under investigation, a lattice gauge theory with finite dihedral group $G=D_4$ as gauge group,
which can be considered as a toy model for more interesting, but harder, systems such as Yang--Mills theories with continuous Lie groups in (3+1) dimensions.

A possible presentation for dihedral groups $D_{n}$ in
terms of two generators $s$ and $r$ (which can be considered respectively as a reflection and a rotation by $2\pi/n$ in a 2-dimensional plane) is the following
\begin{align}\label{eq:Dnpresentation}
D_n =\langle s,r | s^2=r^n=srsr=e \rangle,
\end{align}
where $e$ denotes the identity element.

Since $|D_4|=8$, each link variable can be represented with exactly 3 qubits. As our working basis for the extended Hilbert space, we use the magnetic one, which is defined using the values of the gauge group for each link: 
for a single $D_4$ link variable register, 
its 8 possible states of the computational basis
can then be mapped to group elements as $\ket{x_2,x_1,x_0} \longleftrightarrow s^{x_2} r^{2 x_1+x_0}$,
where $s$ and $r$ are the finite generators appearing in the group presentation of Eq.~\eqref{eq:Dnpresentation} specialized to $n=4$, while $(x_2,x_1,x_0)\in \mathbb{Z}_2^3$ is a triple of binary digits labeling states of the computational basis.

Due to limited resources available, we consider a system with a relatively small lattice with $|V|=2$ vertices and $|E|=4$ link variables associated with the (oriented) lattice edges $E$, namely a $2 \times 1$ square lattice with periodic boundary conditions (PBC) in both directions, as depicted in Fig.~\ref{fig:D4latt}.
\begin{figure}[t]
  \centering
  \includegraphics[width=1.0\linewidth]{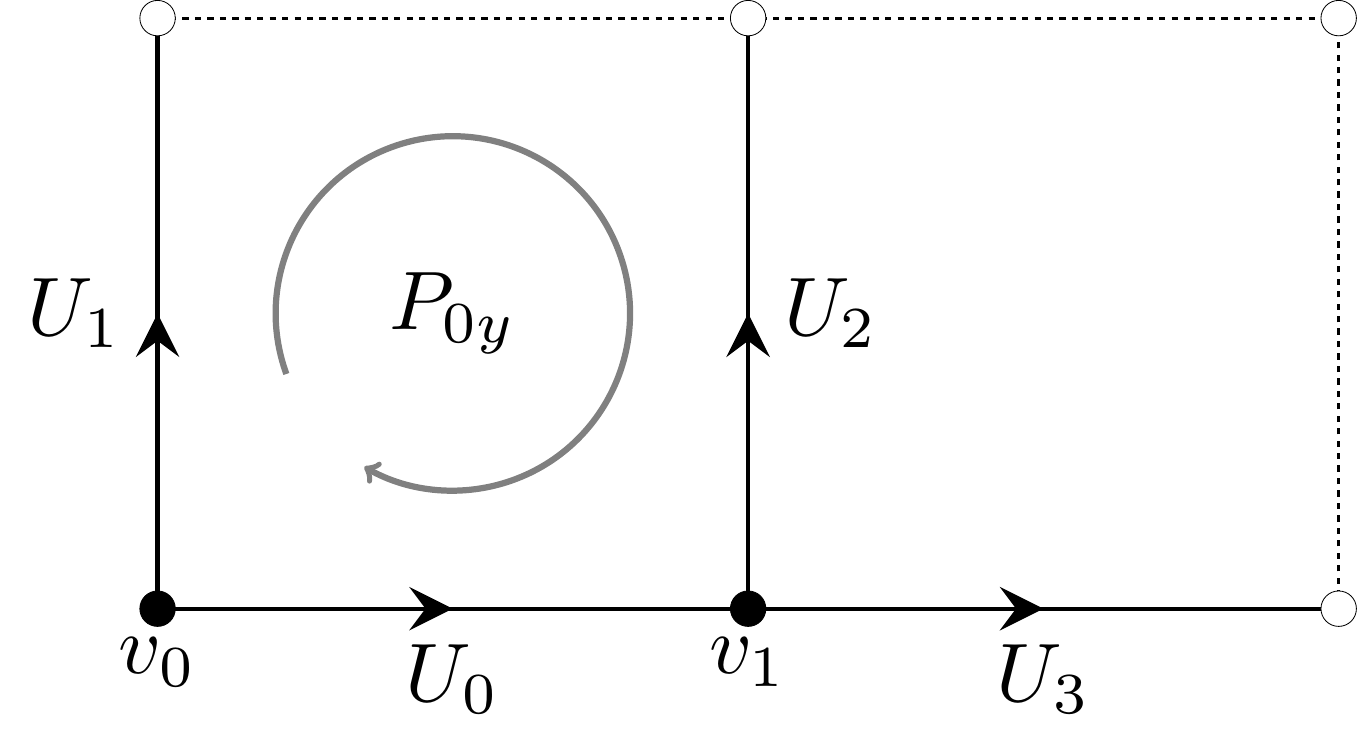}
  \caption{Square lattice with $2 \times 1$ sites and periodic boundary conditions used in this work.
  Vertices are denoted by $v_i$, while link variables are denoted by $U_l$; white dots are identified with black ones, while dashed lines are identified with solid ones on the opposite side.
  The left plaquette is highlighted with a circled arrow and denoted by $P_{0y}$.}
  \label{fig:D4latt}
\end{figure}
Denoting by $\mathcal{H}_{U_l}$ the Hilbert space of each gauge-group-valued variable $U_l$, 
the so-called \emph{extended} Hilbert space 
representing the system can be written as a tensor product \mbox{$\mathcal{H}_{\text{ext}} \equiv \bigotimes\limits_{\{U_l\}} \mathcal{H}_{U_l}=\mathcal{H}_{U_3}\otimes\mathcal{H}_{U_2}\otimes\mathcal{H}_{U_1}\otimes\mathcal{H}_{U_0}$}. 
For later convenience, in the following discussions we use the shorthand \mbox{$\ket{\vec{U}}=\ket{U_3}\otimes\ket{U_2}\otimes\ket{U_1}\otimes\ket{U_0}$} to indicate states of the extended Hilbert space in the computational link basis.

Since the system considered is actually a gauge theory,
only the subspace which is left invariant by the action of arbitrary local
gauge transformations $\mathcal{G}=\{(g_v)\in G^{|V|}\}$ should be considered \emph{physical}.
In terms of link variables, the action of a generic local gauge transformation is 
\begin{align}
 a_{(g_v)}: U_{(v_j \leftarrow v_i)} \mapsto  g_{v_j}^\dagger U_{(v_j \leftarrow v_i)} g_{v_i},
\end{align}
which lifts to a unitary operator acting on the Hilbert space of states as
\begin{align}
 \mathscr{U}_{a_{(g_v)}}=\sum_{\{U_l\}} \ketbra{\{U_{l}\}}{\{g_{v_h(l)}^\dagger U_{l} g^{\phantom{\dagger}}_{v_t(l)}\}},
\end{align}
where $v_t(l)$ and $v_h(l)$ are respectively the \emph{tail} and \emph{head} vertices of the link $l$.
Therefore, physical states are the ones invariant with respect to generic
local gauge transformations, which means 
\begin{align}
    \mathcal{H}_{\text{phys}} \equiv \Inv_{\mathcal{G}}\lbrack \mathcal{H}_{\text{ext}} \rbrack 
    \equiv \bigcap_{(g_v)\in \mathcal{G}}\Ker[\mathscr{U}_{a_{(g_v)}}-\matone].  
\end{align}
While the extended dimension of the system considered is ${|G|}^{|E|}=8^4=4096$ (i.e., 12 qubits on the system register),
as shown in Ref.~\cite{Mariani:2023eix} its physical dimension
can be computed as $\dim \mathcal{H}_{\text{phys}} = \sum\limits_{C \in \text{conj. class}} {\big(\frac{|G|}{|C|}\big)}^{|E|-|V|}$, which is $176$ for our lattice and group choice.
Notice that the left plaquette in Fig.~\ref{fig:D4latt}
\begin{align}\label{eq:P0y}
P_{0y}=\sum\limits_{\{U_l\}} \lbrack U_0^\dagger U_2^\dagger U_0 U_1\rbrack \ketbra{\vec{U}}
\end{align}
 is based on the vertex $v_0$ and cycles in the clockwise direction. This choice turns out to be convenient because we can just use group inversion gates ($\mathfrak{U}_{-1}$) and left group multiplication gates ($\mathfrak{U}_{\times}$),
defined in Ref.~\cite{Lamm:2019bik}, to write the gauge group value of the plaquette on the $U_1$ register without requiring additional ancillary registers, as shown in Fig.\ref{fig:selfplaq1}.
\begin{figure}[t]
    \centering
    \includegraphics[width=0.95\linewidth]{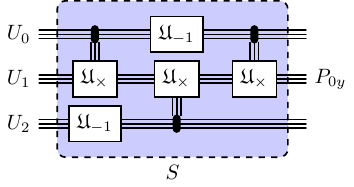}
    \caption{Circuit implementing a change of basis from the basis of the three link registers $U_0$, $U_1$ and $U_2$ in Fig.~\ref{fig:D4latt} to a basis with left plaquette $P_{0y}$ on the central register without further ancillary registers (see Eq.~\eqref{eq:P0y}). The gates $\mathfrak{U}_{-1}$ and $\mathfrak{U}_{\times}$ implement respectively the group inversion and group multiplication (two link registers involved), as defined in Ref.~\cite{Lamm:2019bik}. 
    The same circuit applied to the registers $U_3$, $U_2$ and $U_1$ can be used to rotate into a basis diagonal for the right plaquette operator.}
    \label{fig:selfplaq1}
\end{figure}
The general structure of Hamiltonian that we use in this work is of the Kogut--Susskind form (without matter), i.e., $H = H_V + H_K$, consisting in a magnetic (or potential) term, which encodes the contribution of spatial plaquettes, and an electric (or kinetic) term, which encodes the contribution of timelike plaquettes\footnote{The concept of a timelike plaquette is usually introduced in the standard discretized path-integral formulation obtained after the Trotter-Suzuki decomposition, where also gauge variables corresponding to temporal links are present. Having this in mind, the electric terms can be derived from timelike plaquettes imposing the temporal gauge, $U_{l}=\matone \; \forall l \text{ timelike}$.}. 
Following the same notation as Ref.~\cite{Lamm:2019bik}, for a lattice gauge theory with $D_4$ gauge group these terms can be written as a product over all plaquette terms 
\begin{align}\label{eq:D4_HV}
H_V &= -\frac{1}{g^2} \sum_p \Re\Tr[\prod_{\langle ij \rangle \in p} U_{ij}], \\\label{eq:D4_HK}
H_K &= -\Ln T_K,
\end{align}
where $p$ extends over (path-ordered) plaquettes, $g$ is the coupling parameter of the theory, while $\Ln T_K$ is the matrix logarithm of the kinetic part of the so-called \emph{transfer matrix} $T_K$, defined as having matrix elements 
 \begin{align}
\langle \Vec{U}^\prime | T_K | \Vec{U} \rangle = \prod_{l=0}^{|E|-1} e^{\frac{1}{g^2} \Tr[\rho_f(U^{\prime -1}_l U_l)]},
\end{align}
where $\rho_f$ denotes a fundamental (2-dimensional irreducible) representation of $D_4$.
As discussed in Sec.~\ref{subsec:QMS}, an essential ingredient in the implementation of the Quantum Metropolis Sampling algorithm is the time evolution of the system. 
In this case, this time evolution can be written as a second-order Trotter expansion~\cite{Trotter:1959,Suzuki:1976be} with $N$ time steps:
\begin{align}\label{eq:trot_timevo}
e^{-iHt} \longrightarrow \mathcal{U}(t) = {(e^{-i \frac{H_K t}{2 N}} e^{-i \frac{H_V t}{N}} e^{-i \frac{H_K t}{2 N}})}^N.
\end{align}
In particular, denoting the Trotter step size with $\Delta t = \frac{t}{N}$, the contribution of the potential term 
to the time evolution can be written as
\begin{align}
    e^{-iH_V \Delta t} = \prod_p \mathcal{U}_V^{(1)}(p),
\end{align}
where $\mathcal{U}^{(1)}_V$ acts on a single plaquette $p$. For example, for the left plaquette $P_{0y}$ one has 
\begin{align}\label{eq:pot_term_evo}
\mathcal{U}_V^{(1)}(\text{P}_{0y}) = e^{-\frac{i}{g^2}\Tr \rho_f(\text{P}_{0y}) \Delta t}.
\end{align}
So, in addition to the inversion gate $\mathfrak{U}_{-1}$ and the left group multiplication gates $\mathfrak{U}_{\times}$, a gate implementing the trace of group elements $\mathfrak{U}_{\Tr}$ is required. 
There are only two elements of $D_4$, $e$ and $r^2$, whose trace in the fundamental representation $\rho_f(g)$ is non-zero, i.e., $\Tr \rho_f(e) = 2$ and $\Tr \rho_f(r^2) = -2$. Therefore, one can perform the time evolution for a single plaquette term in Eq.~\eqref{eq:pot_term_evo} by first rotating (using the circuit $S$ as depicted in Fig.~\ref{fig:selfplaq1}) into a convenient basis where the plaquette information is stored in some register $U_l$, then applying a controlled phase gate according to
\begin{align}
    \mathfrak{U}_{\Tr}(\theta)  \ket{x_2,x_1,x_0}_{U_l} = \exp{2 i \theta {(-)}^{x_1}\delta^0_{x_2}\delta^0_{x_0}}\ket{x_2,x_1,x_0}_{U_l},
\end{align}
with $\theta = \frac{1}{g^2} \frac{t}{N}$, and finally rotating back
to the original link basis (using $S^\dagger$, i.e., the inverse circuit of $S$ in Fig.~\ref{fig:selfplaq1}).
By the use of these $3$ gates, named \textit{primitive gates} in Ref.~\cite{Lamm:2019bik}, one can perform the time evolution of the potential term. Now we apply the same idea to the kinetic term $e^{-iH_k\Delta t} = \prod_{\langle i j \rangle}\mathcal{U}_K^{(1)}(i,j)$ . 
In particular, for the 4-link lattice, one can write the kinetic part of the Hamiltonian as
\begin{align}
H_{k} = H_k^{(1)} \otimes \matone \otimes \matone \otimes \matone +   \matone  \otimes H_k^{(1)} \otimes \matone \otimes \matone \\ 
+ \matone  \otimes \matone \otimes  H_k^{(1)} \otimes \matone + 
\matone  \otimes \matone \otimes \matone \otimes  H_k^{(1)}, 
\end{align}
which is the sum of a single variable kinetic Hamiltonian for each link variable of the theory, such that $H_k^{(1)} = -\Ln T_k^{(1)}$, where 
\begin{align}
\bra{U^\prime}T_k^{(1)}\ket{U} = e^{\frac{1}{g^2}\Tr [\rho_f (U^{\prime \dagger}U)]}.
\end{align}
Each $\mathcal{U}_K^{(1)}$ term can be implemented by 
\begin{align}
\mathcal{U}_K^{(1)} (i,j) = \mathfrak{U}_F\mathfrak{U}_{\text{phase}} \mathfrak{U}_F^{\dagger},
\end{align}
where $\mathfrak{U}_F$ is the fourth primitive gate, which performs the Fourier Transform of the $D_4$ group and diagonalizes $T_K^{(1)}$. To implement $\mathfrak{U}_F$, we used the circuit introduced in Ref.~\cite{Lamm:2019bik} (an alternative is reported in Ref.~\cite{alam2022primitive}). 
In addition to each of the $4$ primitive gates, also the $\mathfrak{U}_{\text{phase}}$ depends on the gauge group. 
As discussed with more details in Appendix~\ref{app:transfmat_gropth}, $T_K^{(1)}$
can be written as block diagonal on the basis of irreducible representations (irreps).
Therefore, it becomes diagonal after the application of a Fourier Transform gate $\mathfrak{U}_F$,
while the $\mathfrak{U}_{\text{phase}}$ gate is diagonal and its entries can be computed 
from the contribution of each irrep subspace as reported in~\ref{subapp:transfmatD4}.
Another possibility to define ab initio a kinetic Hamiltonian for finite gauge groups is discussed in Ref.~\cite{Mariani:2023eix}, where the authors point out that there is a certain degree of 
arbitrariness in this definition, which is fixed only by imposing additional physical constraints. 
For example, it is straightforward to show that, 
by requiring Lorentz invariance of the space-time lattice, 
it is possible to match this ab initio definition with 
the one derived from the (Euclidean) Lagrangian formulation, as the one used in Ref.~\cite{Lamm:2019bik} to implement the real-time simulations with $D_4$ group.
More details about this matching can be found in Appendix~\ref{app:transfmat_gropth}.

\section{The algorithm}\label{sec:algo}
Here we give an overview of the algorithm we use to compute 
thermal averages and discuss the specific challenges of its application and 
the adaptations that have to be considered in the case of gauge theories in general, and for the system introduced in Sec.~\ref{sec:sys} in particular.

\subsection{Overview of Quantum Metropolis Sampling}\label{subsec:QMS}
In this Section, we sketch the algorithm we use for the following results. 
This is based on a generalization of the classical Markov Chain Monte Carlo with Metropolis importance sampling called Quantum Metropolis Sampling (QMS)~\cite{QMS_paper}. 
Here we just mention the main features of the QMS which we use in the following
discussion when adapted to the case with gauge invariance. 
A more detailed description of the QMS algorithm and its systematic errors can be found in 
Refs.~\cite{QMS_paper,Qubipf_QMS,Aiudi:2023cyq}.

Given the Hamiltonian representing the system under study, 
one can formally decompose it, according to the spectral theorem, in terms of its spectrum and
eigenspace projectors $H=\sum_{k} E_k \mathbb{P}_{V_k}$. 
The general idea of the QMS algorithm consists of producing a Markov chain of pairs eigenvalue-eigenstates
\begin{align}
    \begin{bmatrix} E_{k_0}\\\ket{\varphi_{0}} \end{bmatrix}
    \to
    \cdots
    \begin{bmatrix} E_{k_i}\\\ket{\varphi_{i}} \end{bmatrix}
    \to
    \begin{bmatrix} E_{k_{i+1}}\\\ket{\varphi_{{i+1}}} \end{bmatrix}
    \to
    \cdots
    \begin{bmatrix} E_{k_M}\\\ket{\varphi_{M}} \end{bmatrix}
\end{align}
where each sampled state belongs to the corresponding $H$ eigenspace (i.e., $\mel{\varphi_{i}}{\mathbb{P}_{V_{k_{i}}}}{\varphi_{i}}=1$). In the QMS algorithm, 
after some number of steps $M$, required for thermalization purposes, 
the probability of sampling a state in $V_k$ reproduces 
the Gibbs weight expected from the density matrix
\begin{align}\label{eq:GibbsWeight}
p_{k} = \Tr[\hat{\rho}(\beta) \mathbb{P}_{V_{k}}]=\mu_k e^{-\beta E_k}/Z 
\end{align}
where $\mu_k= \Tr[\mathbb{P}_k]$ denotes the multiplicity of $E_k$.
Therefore, the terminal eigenstate $\ket{\varphi_{M}}$ of a chain can be used to perform a measurement of the observable one is interested in, 
whose expectation value can then be assembled as a simple average of different chains 
$\expval{\hat{\mathcal{O}}} \simeq \overline{\mathcal{O}}=\frac{1}{N_{\text{chains}}}\sum_{s=1}^{N_{\text{chains}}} O_s$.
Since a random initial state of the extended Hilbert space has almost surely a non-vanishing overlap with the unphysical subspace,
one should explicitly initialize the chain in a gauge invariant way.
Indeed, for any gauge group $G$, it is always possible to initialize
in a gauge invariant state by setting every link variable to the trivial\footnote{If one is interested in studying gauge sectors with non-zero static charges, one can choose as initial state any combination of irreps states which are contained in that sector.}
irrep state (corresponding to its $0$-mode), which is realized by an application of an inverse Fourier transform gate to the zero-mode state for each link register ${\ket{\psi_0}}=\bigotimes_l {\Big(\mathfrak{U}_F^\dagger {\ket{\widetilde{0}}}\Big)}_l$; in the case of the group $D_4$, this initialization is also possible through the application of Hadamard gates
to each link register:
\begin{align}\label{eq:D4init}
    \ket{\psi_0} = \bigotimes_{l\in E}{\Big(\text{Had}^{\otimes 3}\ket{000}\Big)}_l=\bigotimes_{l\in E}{\Big(\frac{1}{\sqrt{8}}\sum_{\vec{x}\in \mathbb{Z}_2^3}{\ket{x_2,x_1,x_0}}_l\Big)}.
\end{align}
The state $\ket{\psi_0}$ obtained is not an eigenstate of $H$, 
but it can be projected into an approximate eigenstate
through the application of a Quantum Phase Estimation (QPE) operator~\cite{nielsen:2010quantum,Kitaev:1995qy}, 
which is described more in detail in Sec.~\ref{subsec:exactqped}.
Indeed, the QPE is one of the main ingredients used to encode and measure 
the energy of states and build the accept-reject oracle of the QMS. 
The unitary operator used for the QPE step is based on a controlled time evolution described 
by Eq.~\eqref{eq:trot_timevo},  whose implementation is discussed in Sec.~\ref{sec:sys}.

Another component of the QMS we need to mention is
the analog of the accept-reject procedure featured in the classical Metropolis algorithm.
This is retrieved by implementing an oracle that makes use of a 1-qubit register,
which we call \emph{acceptance register},
storing the condition for acceptance or rejection~\cite{QMS_paper}, according to 
the Metropolis probability~\cite{Metropolis:1953am} of transition between 
eigenstates given by 
\begin{align}\label{eq:acc_rej}
p_{\text{acc.}}{({\ket{\varphi_{i}}\longrightarrow \ket{\varphi_{j}}})} 
= \text{min}(1,e^{-\beta(E_j-E_i)}).   
\end{align}
However, the quantum nature of the algorithm complicates the rejection process: 
due to the \textit{no-cloning theorem}, after a measurement, 
it is not possible to retrieve from memory the previous state anymore. 
To solve this issue, in Ref.~\cite{QMS_paper} an iterative procedure is proposed, 
whose purpose is to find a state with the same energy as the previous one $E_k = E_i$, 
(i.e., in the same microcanonical ensemble): in this case, 
even if the new state is not exactly the original one, the whole process can be viewed as a
standard Metropolis step, followed by a microcanonical update.
The transition between eigenstates $\ket{\varphi_i} \to \ket{\varphi_{i+1}}$ 
is handled by performing a random choice in a predefined set 
$\mathcal{C}$ of unitary operators called \emph{moves}. 
An application of these, followed by the acceptance oracle, brings any state $\ket{\varphi_i}$ to a superposition of different possible eigenstates, each weighted by an additional contribution from the acceptance probability in Eq.~\eqref{eq:acc_rej}, 
while an energy measurement on this state makes it collapse 
to a specific new eigenstate $\ket{\varphi_{i+1}}$. The new eigenstate is then accepted or rejected according to a measurement on the acceptance register.
As discussed in Sec.~\ref{subsec:GImoves} in more detail, in the case of gauge theories
one should also guarantee that the choice and implementation of moves
preserves gauge invariance of the trial state.

Furthermore, there is another issue due to the quantum nature of the algorithm, 
emerging when one needs to compute thermal averages of any 
(gauge invariant) observable $\hat{O}$ not commuting 
with the Hamiltonian ($\lbrack\hat{O}, \hat{H}\rbrack \neq 0$). 
Measuring such observables will make the state in the system register collapse 
to an eigenstate of $\hat{O}$ which, 
in general, does not belong to any eigenspace of $\hat{H}$, 
bringing the Markov chain out of thermodynamic equilibrium after measurement. 
Since any measurement should be performed only once the Markov chain 
is stationary (or \emph{thermalized}), there are at least two ways to reach this goal: 
the simplest one consists of resetting the Markov Chain, 
initializing the system state again as in Eq.~\eqref{eq:D4init} 
and starting over with a new chain; 
another possibility consists of measuring again the energy of the resulting state 
and performing a certain number $r$ of QMS step to make the chain \emph{rethermalize} before a new measurement. 
As argued in Ref.~\cite{Qubipf_QMS}, 
the latter approach has the advantage of starting from a state 
that has a higher overlap with the stationary distribution.
However, using this approach in the case of gauge theories, 
we need to ensure the additional requirement of preserving gauge invariance 
of the states at the measurement stage, as described in Sec.~\ref{subsec:retherm-GImeas}.

The code of the QMS algorithm with $D_4$ lattice gauge theory, 
used to obtain the results for this paper, 
is implemented with a hybrid quantum-classical emulator developed 
by some of the authors and publicly available in~\cite{SUQApub}.

\subsection{Gauge invariant ergodic moves}\label{subsec:GImoves}
In the case of gauge theories, the algorithm discussed in the previous section needs to be adapted in order to ensure that the state is initialized and maintained as a gauge invariant state.
This means that the moveset $\mathcal{C}=\{U_m=e^{i \theta_m A_m}\}$ should satisfy both \emph{gauge invariance} and \emph{ergodicity}.
The former requirement can be satisfied by using only gauge invariant generators $\{A_m\}$, i.e., 
such that $[A_m,\mathcal{G}_v]=0 \quad \forall m,v$.
The condition of ergodicity, in the case of QMS, means that the action of an arbitrary sequence of moves is transitive in the space of physical states (i.e., gauge invariant states).
This guarantees that the whole physical Hilbert space is in principle within reach, but it does not give information about efficiency. 
As in the case of classical Markov Chain Monte Carlo with importance sampling, the possibility of reaching any possible physical state does not 
in general correspond to a uniform exploration (unless the system is studied in the extremely high-temperature regime, i.e., vanishing $\beta$ as studied in Sec.~\ref{subsec:numres-beta0}), 
so the curse of dimensionality becomes treatable.

In order to guarantee an ergodic exploration of the whole physical Hilbert space, 
we exploit the property that a set $\mathcal{C}=\{U_m\}$ made of moves generated using two random hermitian and gauge invariant generators is sufficient to explore the whole physical Hilbert space (see Appendix~\ref{app:twogens} for a more precise statement and a sketch of the proof).
This allows us to just use two random gauge invariant Hermitian operators
as infinitesimal generators of the special unitary group $SU(\mathcal{H}_{\text{phys}})$. 
At this point, there is some freedom in this random selection 
but, in practice, we make a specific choice that makes use of a useful partition of the generators inspired by the proof of Theorem~\ref{thm:twoDenseSUN}.
First of all, we notice that it is possible to associate the set 
of all gauge invariant Hermitian operators that are diagonal in link basis\footnote{In Ref.~\cite{Durhuus:1980}) has been proved that, for Lie groups, the set of all Wilson lines (as gauge invariant Hermitian operators) is sufficient to span the whole space of gauge invariant functions in link basis, but this is not guaranteed to hold for some gauge theories with finite gauge group.} to a Cartan subalgebra of $SU(\mathcal{H}_{\text{phys}})$. 
For the same reasons, the Hermitian operators corresponding to projectors 
into different irreps for individual link variables 
(appearing as the generators of the kinetic part of the transfer matrix for each link),
can be associated to the roots elements of the algebra.
According to the \emph{root space decomposition}~\cite{georgi1982lie}, 
these two (mutually non-commuting) sets of Hermitian operators form a basis
for the full algebra. Therefore, by Theorem~\ref{thm:twoDenseSUN} 
the two generators, built as random linear combinations of elements from both sets, 
generate the whole (special) unitary group $SU(\mathcal{H}_{\text{phys}})$.
In the case of our system, in practice, we considered a random linear combination 
$\hat{A}_1 \equiv \sum_{\gamma} r^{(1)}_{\gamma} \hat{W}_\gamma$ 
of independent Wilson line operators to define the generator for the first move $R_1 = e^{i \theta_1 \hat{A}_1}$, and a random linear combination $\hat{A}_2 \equiv \sum_{l,j} r^{(2)}_{j;l} \mathbb{P}_j^{(l)}$ of irrep projectors (defined in Appendix~\ref{app:transfmat_gropth} for each link variable $l$) to define the generator for the second move $R_2 = e^{i \theta_2 \hat{A}_2}$.
This set of moves $\mathcal{C}\equiv\{R_1,R_1^\dagger,R_2,R_2^\dagger\}$
is \emph{ergodic} in the sense of allowing the reachability of all physical eigenstates 
after a finite sequence of applications. 
This is theoretically guaranteed with probability 1 by Theorem~\ref{thm:twoDenseSUN}, 
and checked numerically in Sec.~\ref{subsec:numres-beta0}.
At this level, we did not mention arguments about efficiency, 
but, for practical considerations,
the quality of numerical results have not shown particular improvements with different choices
of the moveset.

As a final consideration, we should mention that the procedure of writing down 
all independent Wilson lines (i.e., products of link variables on closed loops) 
used in the construction of one of the generators does not scale well for larger lattices.  
One possibility, instead of precomputing the action of all of Wilson loops, 
is to build random closed loops with arbitrary length,
possibly with an exponential tail in the random distribution preventing them from diverging 
in practice, and using them to build generators on the fly. 
This would in general require more steps to allow arbitrary overlaps with physical states,
but it would be manageable in principle.
Another possibility would be to formally identify how a generic transformation
of a link variable acts on neighboring link variables, 
i.e., the ones that share the same vertices in the lattice, 
and considering the projection of the output state that preserves the gauge sector.
In principle, this can be done by a generic transformation on the link variable register, 
followed by a measurement of the Gauss law operators 
(or of the gauge invariant projectors for each vertex, in the case of finite groups).
However, in this case, there would be a non-negligible probability 
that the projective measurement yields an unphysical state, forcing the whole chain to be restarted. 

\subsection{Effects of Quantum Phase Estimation on measured spectrum}\label{subsec:exactqped}
One of the most significant sources of systematic errors in QMS is due to the 
Quantum Phase Estimation (QPE) step, 
used to estimate the energy of the system state. 
The operator $\Phi_{\text{QPE}}$ has the effect of ``writing'' 
an estimate of the eigenvalue $E_k$ associated with the eigenvector $\phi_k$
on the energy register, which is represented by $q_e$ qubits:
\begin{align}\label{eq:QPE_Phi}
    \Phi_{\text{QPE}}: \ket{0}^{\otimes q_e} \ket{\phi_k} \longmapsto \ket{E_{k}} \ket{\phi_k}.
\end{align}
In practice, this is done by defining a uniform grid in the range between
 some chosen $E_{\text{min}}^{(\text{grid})}$ and $E_{\text{max}}^{(\text{grid})}$, which are respectively mapped to the states $\ket{00\dots0}$ and $\ket{11\dots1}$ of the computational basis for the energy register. 
 The other states correspond to the grid sites 
\mbox{$E_j^{(\text{grid})}=E_{\text{min}}^{(\text{grid})}+\varepsilon j$} for all $j=0,\dots,2^{q_e}-1$, and with uniform grid spacing $\varepsilon = \frac{E_{\text{max}}^{(\text{grid})}-E_{\text{min}}^{(\text{grid})}}{2^{q_e}-1}$. 
In the following discussions, we refer to these levels as the \emph{QPE grid}.
Besides very special cases, the energy levels of the system $E_k$ do not fit with the sites of the grid, so we need to take care of QPE.
This means also that the individual states in the chain would only be \textit{approximate} eigenstates, making the actual exact spectrum be \emph{distorted} by the presence of the QPE grid. 
The effect of this kind of error coming from a finite energy resolution used 
in the accept-reject stage has been investigated in the case 
of classical Markov chains (see Ref.~\cite{Roberts:1998,Breyer:2001}).
As discussed in Refs.~\cite{QMS_paper,benenti2004principles,Cleve:1998}, 
the squared amplitude of the $j-$th state of the grid 
for a QPE applied to a system state with true energy $E_k$ is
\begin{equation}\label{eq:QPEcoeffs}
\begin{aligned}
    {|c_{k,j}|}^2 &= \frac{1}{4^{q_e}} \frac{\sin^2\big[\frac{\pi}{\varepsilon} (E_k-E_j^{(\text{grid})})\big]}{\sin^2 \big[\frac{\pi}{\varepsilon 2^{q_e}} (E_k-E_j^{(\text{grid})})\big]},
\end{aligned}
\end{equation}
which is peaked around the true eigenvalue $E_k$. 
Since the energy measurements of the QMS lie on the QPE grid, 
the whole energy distribution sampled would be affected by this \emph{QPE distorsion}, but we still need to take into account the different contribution associated with each eigenvalue coming from the Gibbs weights in Eq.~\eqref{eq:GibbsWeight}.
Having access to the exact spectrum and energy distribution, 
we can write a rough estimate of the QPE-distorted energy distribution 
using the coefficients in Eq.~\eqref{eq:QPEcoeffs} to determine a ``distorsion'' map
\begin{align}\label{eq:grid_therm}
E_k &\longrightarrow \widetilde{E}_k  = \sum_j \abs{c_{k,j}}^2 E^{(\text{grid})}_j,\\
\label{eq:grid_therm_c}
w_k = \frac{e^{-\beta E_k}}{Z} &\longrightarrow \tilde{w}_k = \sum_k \frac{e^{-\beta \tilde{E}_k}}{\tilde{Z}}.
\end{align}
The first equation describes the mean value of the energy applying a QPE to an exact eigenstate with energy $E_k$, according to Eq.~\eqref{eq:QPEcoeffs}. 
The weights in Eq.~\eqref{eq:grid_therm_c} can then be assembled to build 
the expected QPE-distorted distribution as follows:
\begin{equation}\label{eq:QPE-distorsion}
    p^{(QPEd)}_j = \sum_k \abs{c_{k,j}}^2 \tilde{w}_k.
\end{equation}
At the same time, the expectation value of the energy is expected 
to be represented more accurately as follows:
\begin{equation}\label{eq:QPE-distorsionExpval}
\langle H \rangle = \sum_j E_j w_j \longrightarrow {\langle H \rangle}^{(QPEd)} = \sum_k \tilde{E}_k \tilde{w}_k.
\end{equation}
Actually, the distortion described by Eq.~\eqref{eq:QPE-distorsion} 
is not the only effect of the introduction of a QPE grid: 
indeed, for a step $\ket{\varphi_i}\longrightarrow \ket{\varphi_j}$ of the Markov Chain, 
the acceptance probability $p_{\text{acc.}}$ in Eq.~\eqref{eq:acc_rej} involves
states which are not exactly eigenstates of the Hamiltonian, but a superposition of
them.
In particular, as shown by the behavior of numerical results presented in Sec.~\ref{sec:numres},
Eq.~\eqref{eq:QPE-distorsion} is speculative and inaccurate in some regimes, 
and it appears to represent well the QMS data only at small values of $\beta$.
Understanding how the QPE affects the statistical weights is not a trivial problem. 
However, as argued in Sec.~\ref{subsec:numres-betafin}, this source of systematic error is expected to disappear 
as the number of qubits in the energy register is increased and the artifacts of the finite representation become negligible.

\subsection{Revert procedure and tolerance}\label{subsec:algo-revert}
As mentioned in~\ref{subsec:QMS}, when the chain step $i \longrightarrow j$ 
is rejected, the limitations imposed by the \textit{no-cloning} theorem 
are overcome by means of an iterative procedure, 
which stops whenever a measurement yields the same energy as the previous state: $E_j^{(\text{grid})} = E_i^{(\text{grid})}$. 
In general, it is useful to set a maximum number of iterations for this procedure, 
after which the Markov Chain is aborted and the system state is initialized again, requiring a new thermalization stage. 
This occurrence affects the efficiency of the algorithm, 
reducing the typical length of allowed Markov chains and slowing down the general sampling rate of observables measured, 
and becomes worse as the number of qubits $q_e$ in the energy register and the inverse temperature $\beta$ are increased.
Indeed, the higher $\beta$, the more likely a move step is rejected due to a lower acceptance probability.
At the same time, the higher $q_e$, 
the more difficult is to revert back to the same site of the grid that corresponds 
to the previous state $E_i^{(\text{grid})}$, 
i.e. the average number of needed iterations increases dramatically. 
To this end, we propose a simple solution, consisting a relaxation of the constraint 
$E_j^{(\text{grid})} = E_i^{(\text{grid})}$ to a more easily 
achievable condition $\abs{E_j^{(\text{grid})}-E_i^{(\text{grid})}} \leq \varepsilon m_{\text{tol}}$, 
where $m_{\text{tol}}$ represents the accepted tolerance in grid units, 
while $\varepsilon = \frac{E_{\text{max}}^{(\text{grid})}-E_{\text{min}}^{(\text{grid})}}{2^{q_e}-1}$ is the grid spacing. 
Notice that at $\beta=0$ all moves are always automatically accepted, 
since the Metropolis acceptance probability Eq.~\eqref{eq:acc_rej} is one for transition 
between any eigenstates. Therefore, there is no need to use a non-zero tolerance in this case (i.e., $m_{\text{tol}}=0$ for $\beta = 0$).

\subsection{Rethermalization and gauge invariant measurements}\label{subsec:retherm-GImeas}
As mentioned in Sec.~\ref{subsec:QMS}, when a rethermalization strategy is used,
one should perform measurements in a gauge invariant fashion, 
in order to ensure gauge invariance of the state even after measurement.
Since any gauge invariant observable $\hat{O}$ 
(such as the trace of the real part of a plaquette operator)
commutes with a generic local gauge transformation $\mathcal{G}$, 
each of its eigenspaces must also commute, i.e., 
$\hat{O} = \sum_s \lambda_s \mathbb{P}_{V_s}$ and $[\mathcal{G},\mathbb{P}_{V_s}]=0$,
where we denote the eigenvalues and eigenspaces of $\hat{O}$ by pairs of $(\lambda_s,V_s)$,
while $\mathbb{P}_{V_s}$ is the projector operator into the eigenspace $V_s$.
A proper gauge invariant measurement should then project into these eigenspaces  
(or any unions of subsets of them), otherwise there is no guarantee that the collapsed state after measurement would be gauge invariant.
For example, in the case of a measurement of the trace of a plaquette with $D_4$ gauge group, the possible values observed are $2^{(1)}$, $-2^{(1)}$ and $0^{(6)}$ (with multiplicity of the eigenspaces shown in parenthesis). 
Once the product of link variables composing a plaquette $\text{Pl}$ is stored in a gauge group-valued register, the corresponding eigenspaces of the real part of the trace can be expressed using projectors which are diagonal in the magnetic basis, introduced in Sec.~\ref{sec:sys}, as follows:
\begin{align}
    \Re \Tr \rho_f (\text{Pl}) &= (+2) \mathbb{P}_{V_{+2}}+(-2) \mathbb{P}_{V_{-2}} + (0) \mathbb{P}_{V_{0}}\\
                         &= +2 \ketbra{e}{e}_{\text{Pl}} -2 \ketbra{r^2}{r^2}_{\text{Pl}},
\end{align}
where $\ket{\cdot}_{\text{Pl}}$ denotes a plaquette register and $\rho_f(\cdot)$ is a fundamental representation of $D_4$.
If we directly measured the state on the three qubits representing the left plaquette $P_{0y}$,
we would get the correct eigenvalue (either $+2$, $-2$ or $0$), 
but the state resulting from the collapse would in general not be gauge invariant anymore (at least, not if the measurement returns $V_0$, whose multiplicity is $6$).

In other words, one should always be careful not to export, naively, classical computational schemes which are not suitable to a quantum context.
In classical simulations of lattice gauge theories, it is usual to write numerical codes  
which go through the computation of non-gauge-invariant quantities before obtaining the desired gauge invariant observable; for instance,
like in this particular case, closed parallel transports are first computed, which are not gauge invariant and transform in the adjoint representation, taking their gauge invariant trace thereafter. This computational scheme does not work in this context, at least if one wants to keep a gauge invariant 
physical state through all the steps of the quantum computation.

Instead, in order to keep gauge invariance of the resulting state, 
we can first perform a measurement discriminating between $V^\prime\equiv V_{+2}+V_{-2}$ and $V_0$ and then, conditionally to the results, if a collapse into $V^\prime$ happens, another measurement is done to discriminate between $V_{+2}$ and $V_{-2}$.
This measurement procedure is sketched in Fig.~\ref{fig:GImeasD4}
and described in detail in the caption.
Notice that the terminal state is collapsed to an eigenstate of the (real part of the trace of the) plaquette, but unlike destroying the state after measurement, we can continue using it as a starting point for rethermalization.
\begin{figure*}
    \centering
    \includegraphics[width=0.95\linewidth]{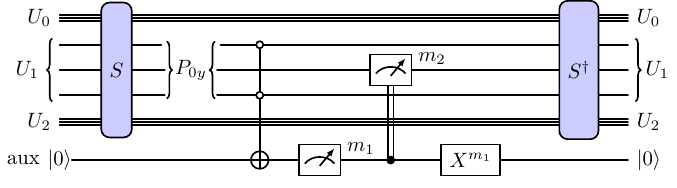}
    \caption{Hybrid protocol implementing a gauge invariant measurement for the plaquette value $\Re \Tr \rho_f(P_{0y})$.
    The $S$ gate group is implemented with the circuit shown in Fig.~\ref{fig:selfplaq1}. The result of the first measurement might yield either $m_1=0$ or $m_1=1$: in the first case, the terminal state is projected into the $V_0$ eigenspace of $\Re \Tr \rho_f(P_{0y})$ and nothing else has to be done; in the second case (i.e., for $m_1=1$), a further measurement on the $P_{0y}$ register yields either $m_2=0$ or $m_2=1$, while the state is projected respectively on either $V_{+2}$ or $V_{-2}$.
    The auxiliary register can then be reset to $0$ according to the result of the first measurement $m_1$ with a flip gate.}
    \label{fig:GImeasD4}
\end{figure*}

\section{Numerical Results}\label{sec:numres}

In this Section we are going to illustrate the numerical results obtained 
for the quantum simulation, through the QMS algorithm discussed in the previous Section, of the thermal ensembles
of the pure-gauge $D_4$ lattice gauge theory with topology depicted in Fig.~\ref{fig:D4latt}.
We will discuss, in particular, the sampled distribution over the Hamiltonian eigenvalues,
comparing it with theoretical predictions, as well as the average energy and plaquette.
 
For the purpose of studying mainly the gauge adaptation of this algorithm 
(and not the physics of the system), without loss of generality, 
in the following discussion we use the Hamiltonian made of the terms~\eqref{eq:D4_HV} and~\eqref{eq:D4_HK}, always fixing the gauge coupling to the value $\frac{1}{g^2}=0.8$,
which results in a spectrum well spread between $E_{\text{min}}^{(\text{phys})}\simeq -11.172$ 
and $E_{\text{max}}^{(\text{phys})}\simeq -1.998$.
In order to prevent leak effects on the boundary of the QPE grid range (see discussion in Sec.~\ref{subsec:exactqped}), we made a common conservative choice of the range for all 
the number of qubits for the energy register investigated ($q_e=3,\dots,7$), 
namely $[E_{\text{min}}^{(\text{grid})},E_{\text{max}}^{(\text{grid})}]=[-13,0]$.
The systematic error coming from a finite Trotter size has been assessed 
and, for the following results, 
we found it to have negligible effects on 
the spectrum distribution for $N=10$ time steps for each power of the time evolution operator in the QPE (i.e.,
$\delta t \sim \frac{\pi (1-2^{-q_e})}{10 \cdot \Delta E^{(\text{grid})}}$).

The QMS has been implemented based on the set of gauge-invariant ergodic moves illustrated in the 
previous Section, namely, $\mathcal{C}\equiv\{R_1,R_1^\dagger,R_2,R_2^\dagger\}$, assigning an equal 25~\% probability 
of selecting one of the 4 moves at each step.
Furthermore, as discussed in Sec.~\ref{subsec:algo-revert},
to gain better efficiency, at the cost of losing some resolution in energy, 
it is useful to set a tolerance margin $m_{\text{tol}}=3$ for the revert procedure whenever 
a move is rejected, which happens more frequently at higher values of $\beta$. 
For our results at $\beta=0.5$, we use $m_{\text{tol}}=3$ for $q_e \geq 5$.

\subsection{Tests of ergodicity and gauge invariance}\label{subsec:numres-beta0}
We first consider the case at infinite temperature ($\beta=0$), 
which would ideally result in a uniform sampling of the whole physical Hilbert space (i.e., $\rho(\beta=0) = {|{\dim \mathcal{H}_{\text{phys}}}|}^{-1}\mathbb{P}_{\mathcal{H}_{\text{phys}}}$).
Therefore, a proper sampling of this distribution would serve both as a check of gauge invariance and ergodicity. Indeed, no unphysical energy levels should be detected and all eigenspaces of $H$ should be explored with the correct physical multiplicity $\mu^{(\text{phys})}_k = \dim {(V_k \cap \mathcal{H}_{\text{phys}})}$.

In the present work, we use two approaches to represent the energy distributions for the exact, QPE-distorted, and numerical data:
on one hand, we make histograms on bins around QPE grid points and with bin size corresponding 
to grid spacing; on the other hand, given the different domains between the exact spectrum and 
the one measured on the QPE grid, we perform a smoothing of the distributions using the Kernel Density Estimation (KDE) technique. More details on such technique are illustrated and discussed in Appendix~\ref{app:GKDE}.
Fig.~\ref{fig:data_b0_r1_th50_ne357} shows the energy distributions measured at $\beta=0$ (i.e., the whole physical spectrum) for different numbers of qubits in the energy registers, using both a histogram representation with bins centered on the QPE grid sites and a KDE representation with smoothing parameter of the KDE kernel functions set to match the bin size of the histograms, i.e., $\sigma_{\text{KDE}}=\frac{\Delta E^{(\text{grid})}}{2^{q_e}-1}$.
The energy distribution distorted by QPE as described in Sec.~\ref{subsec:exactqped} and the one of the exact spectrum are also shown in comparison.
\begin{figure*}
\begin{subfigure}{0.475\linewidth}
  \includegraphics[width=0.95\textwidth]{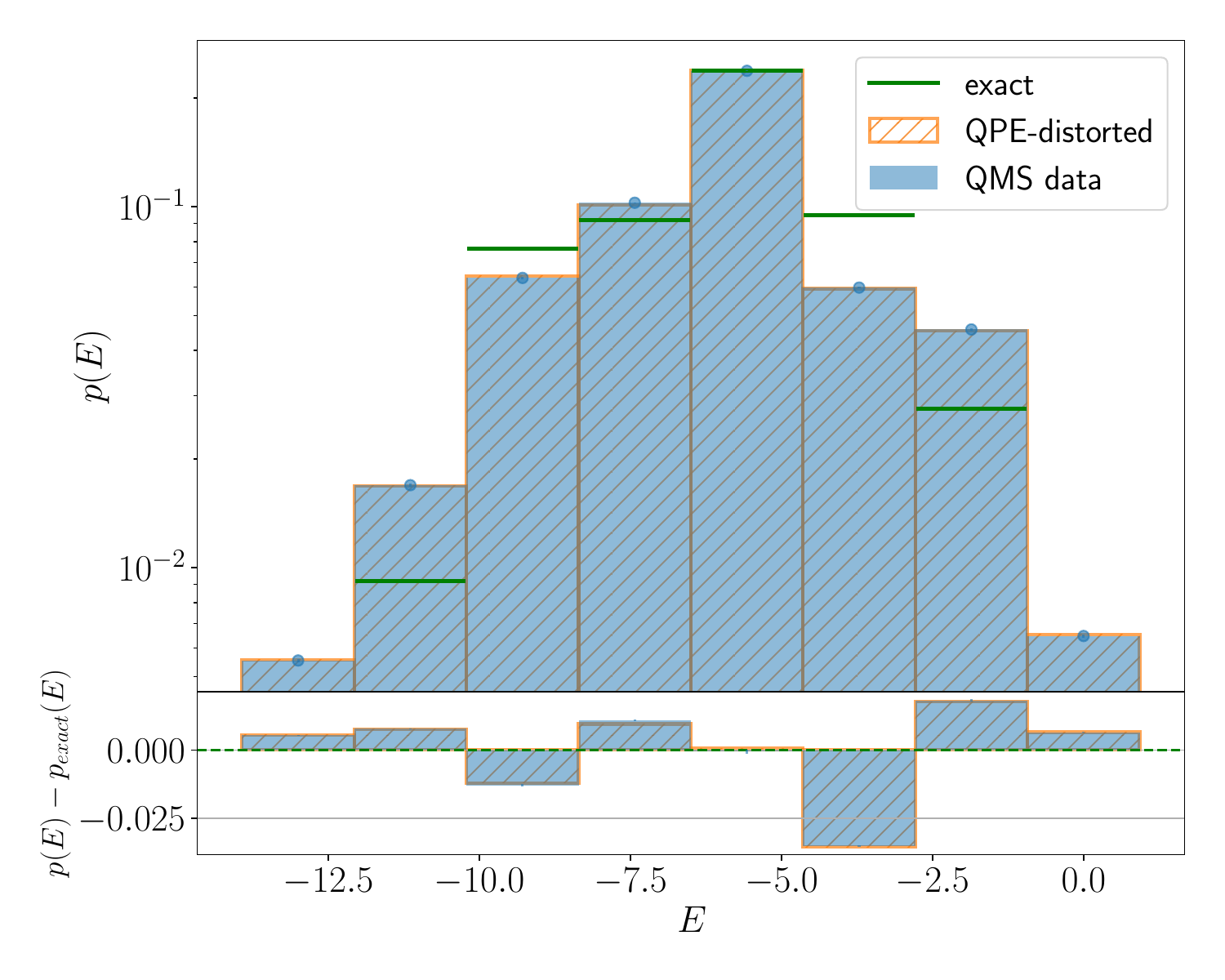}
  \caption{Binned histogram of the energy distribution ($q_e=3$).}
  \label{fig:histo_b0_r1_th50_ne3}
\end{subfigure}\hfill
\begin{subfigure}{0.475\linewidth}
  \includegraphics[width=0.95\textwidth]{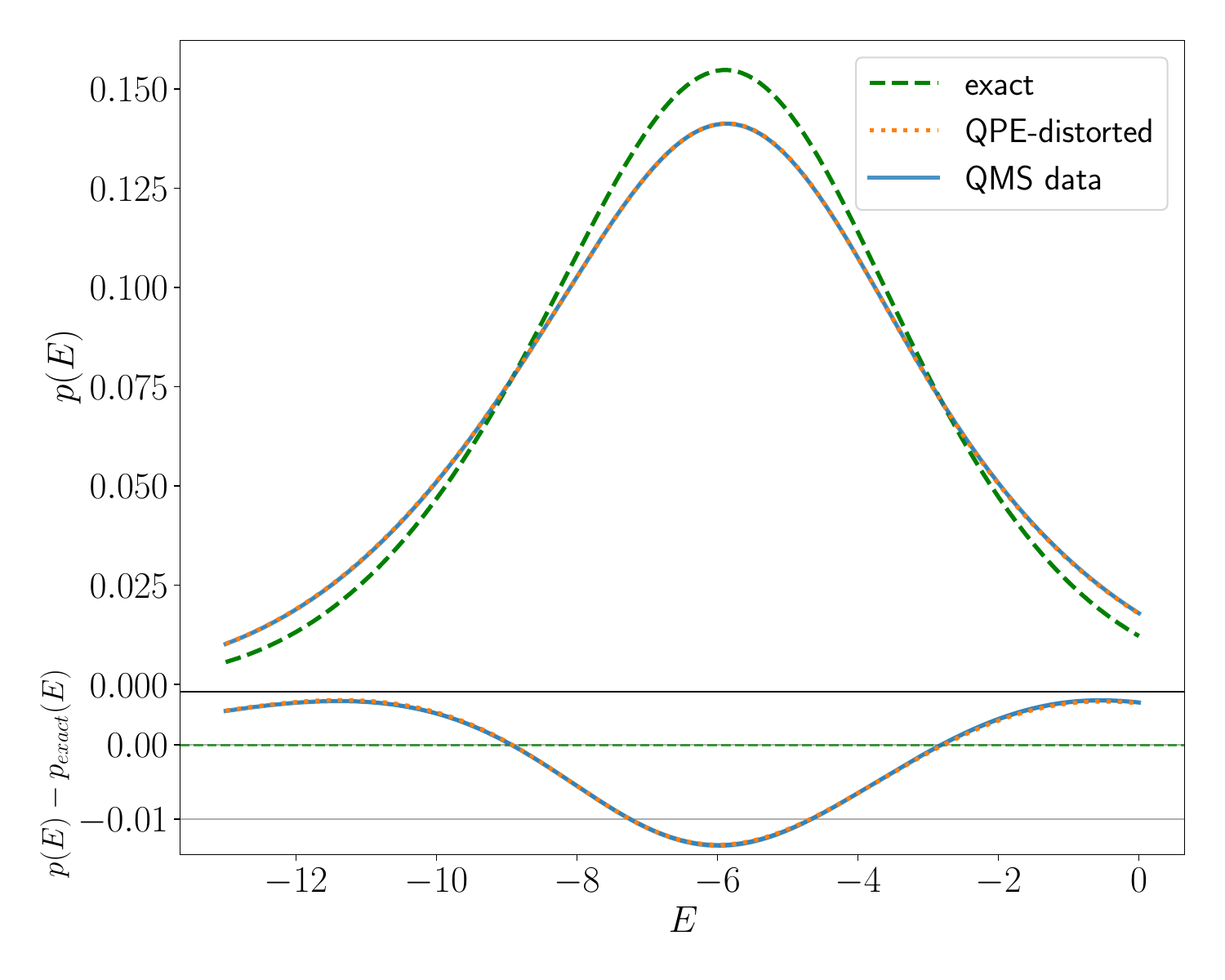}
  \caption{KDE of energy distribution ($q_e=3$, $\sigma_{\text{KDE}}\simeq 1.86$).}
  \label{fig:kde_b0_r1_th50_ne3}
\end{subfigure}
\begin{subfigure}{0.475\linewidth}
  \includegraphics[width=0.95\textwidth]{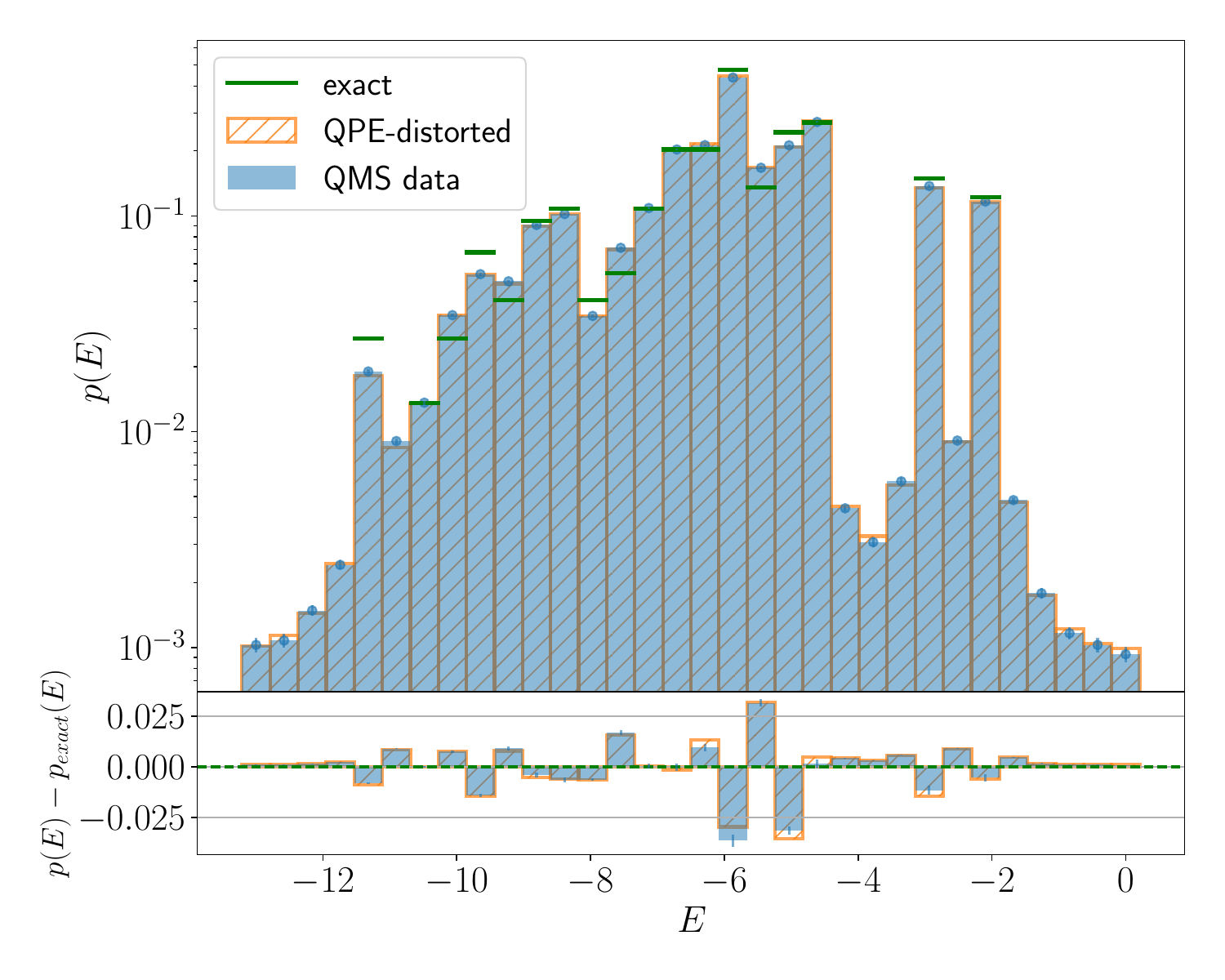}
  \caption{Binned histogram of the energy distribution ($q_e=5$).}
  \label{fig:histo_b0_r1_th50_ne5}
\end{subfigure}\hfill 
\begin{subfigure}{0.475\linewidth}
  \includegraphics[width=0.95\textwidth]{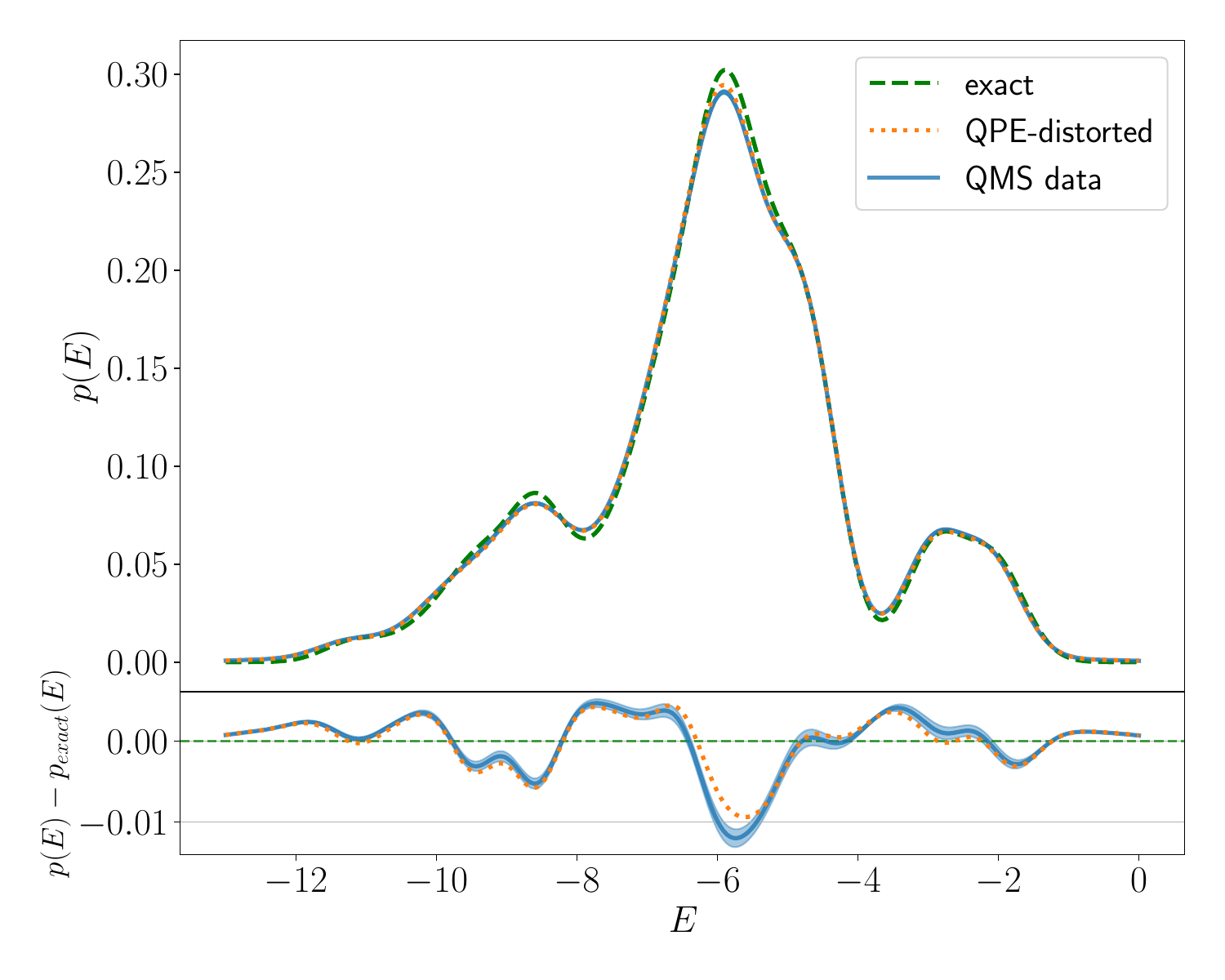}
  \caption{KDE of energy distribution ($q_e=5$, $\sigma_{\text{KDE}}\simeq 0.42$).}
  \label{fig:kde_b0_r1_th50_ne5}
\end{subfigure}

\begin{subfigure}[b]{.475\linewidth}
  \includegraphics[width=0.95\textwidth]{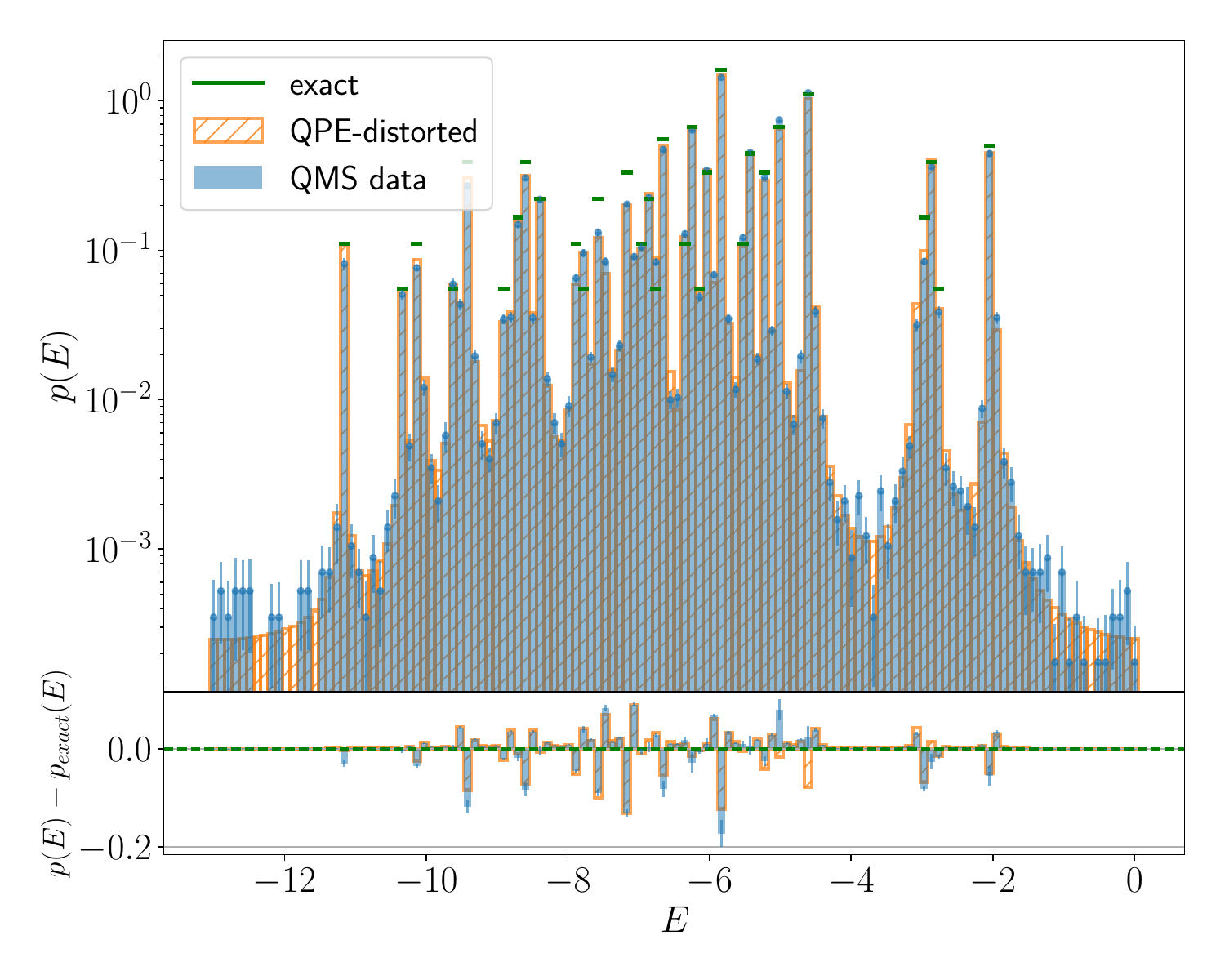}
  \caption{Binned histogram of energy distribution ($q_e=7$).}
  \label{fig:histo_b0_r1_th50_ne7}
\end{subfigure}\hfill
\begin{subfigure}[b]{.475\linewidth}
  \includegraphics[width=0.95\textwidth]{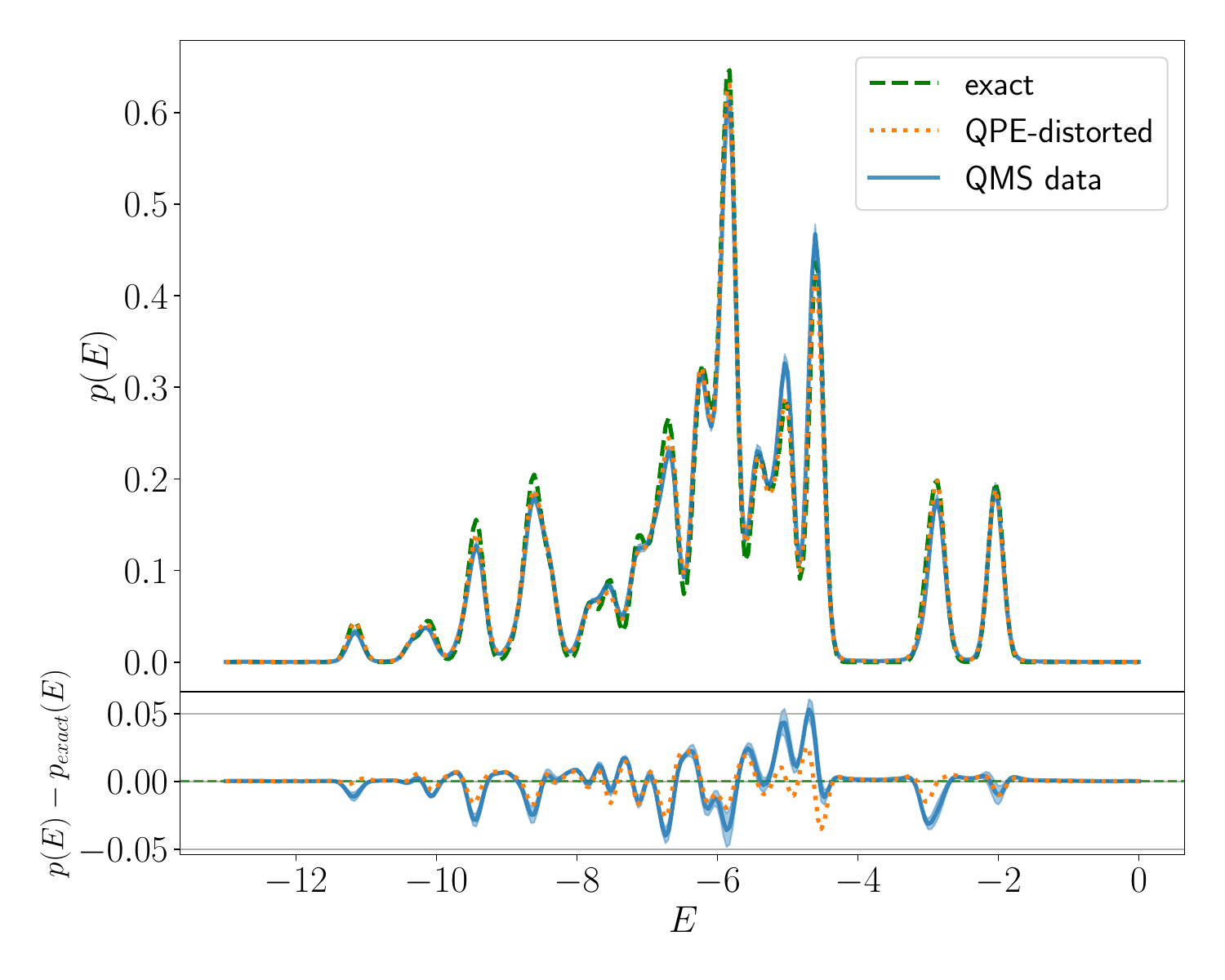}
  \caption{KDE of energy distribution ($q_e=7$, $\sigma_{\text{KDE}} \simeq 0.10$).}
  \label{fig:kde_b0_r1_th50_ne7}
\end{subfigure}

\caption{Histogram and KDE representations of the energy distribution at $\beta=10^{-7}$ for the exact spectrum, the spectrum distorted as expected by QPE (see Eq.~\eqref{eq:QPE-distorsionExpval}) and the data measured via QMS using $q_e = 3$, $5$ and $7$ qubits for the energy register and about 834k, 439k and 56k measurement samples. QMS data has been obtained using 1 rethermalization step (without plaquette measurement) and $50$ thermalization steps, with errors estimated via blocking and bootstrap resampling. 
}
\label{fig:data_b0_r1_th50_ne357}
\end{figure*}

We can investigate more precisely the discrepancy between the exact distribution, the one expected from
the exact one distorted by QPE onto the measurement grid, and the measured data via QMS,
by computing the cumulative distribution. The result of this is shown in Fig.~\ref{fig:data_b0_r1_th50_ne7_KS}
for $q_e=7$ qubits, which is the case that most accurately represents the exact results.
\begin{figure*}
\begin{subfigure}{0.475\linewidth}
  \includegraphics[width=0.95\textwidth]{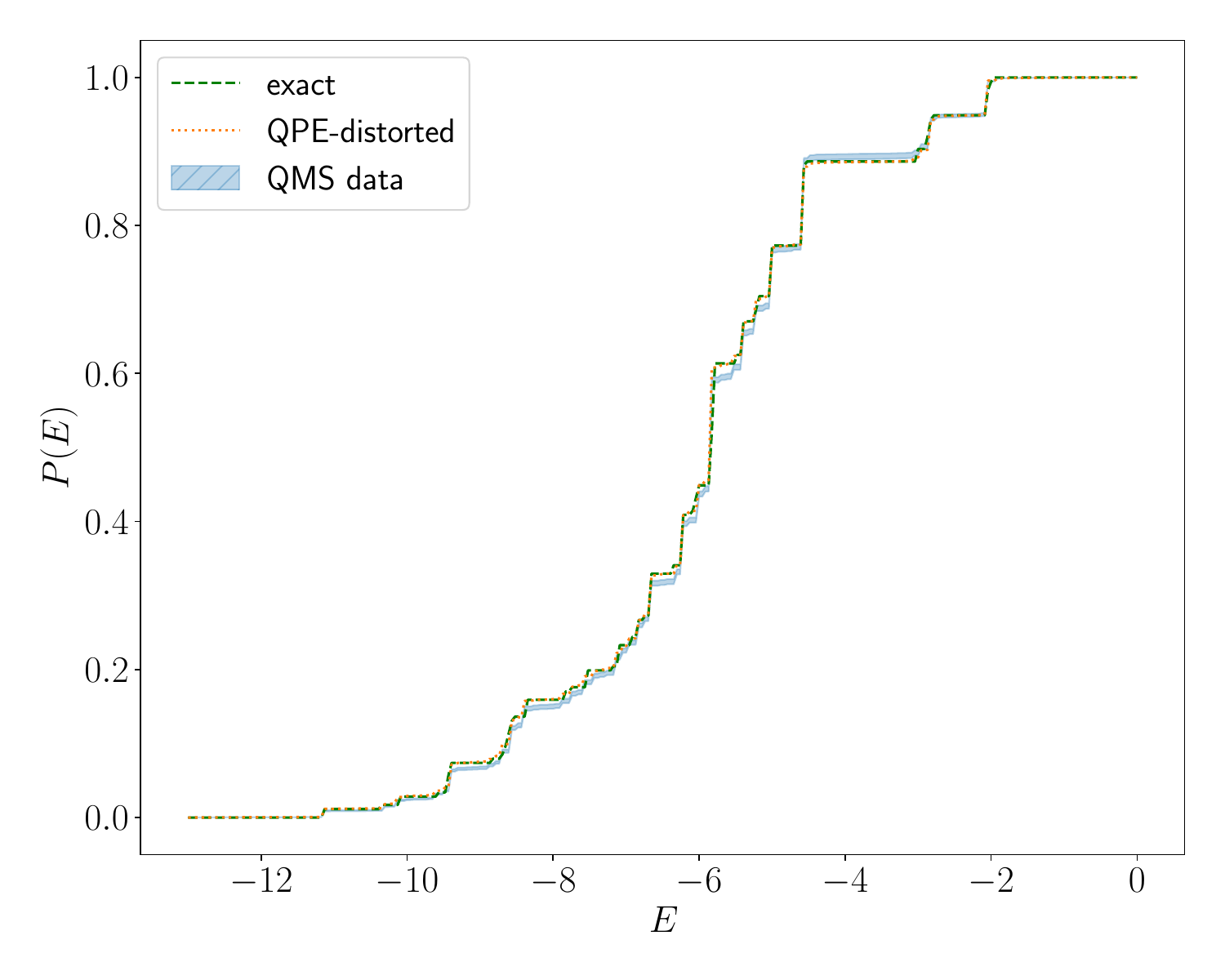}
\caption{Cumulative energy distributions.}
\label{fig:KS_distrib_b0_r1_th50}
\end{subfigure}
\begin{subfigure}{0.475\linewidth}
  \includegraphics[width=0.95\textwidth]{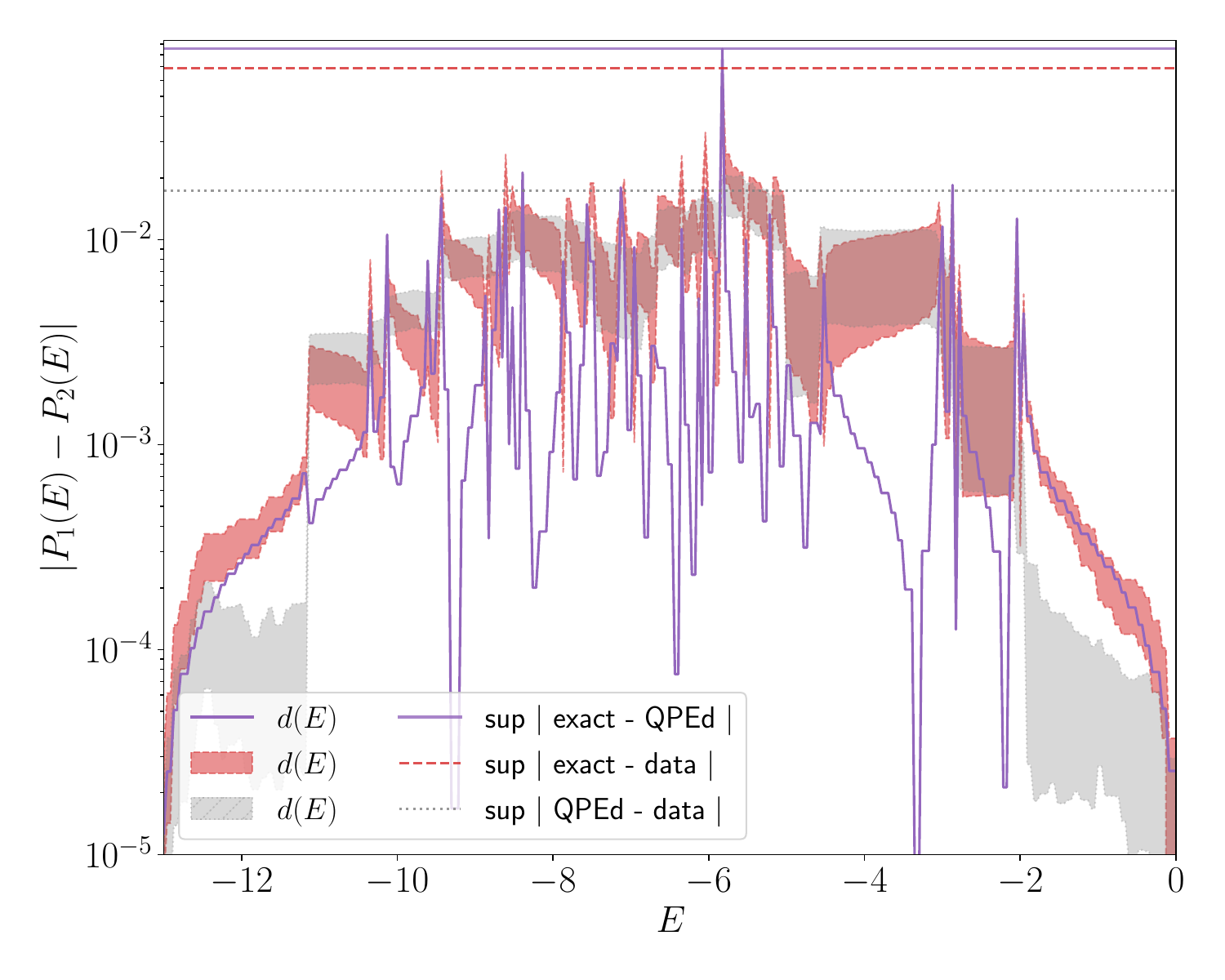}
\caption{Pointwise distance of cumulative energy distributions.}
\label{fig:KS_profile_b0_r1_th50}
\end{subfigure}\hfill
\caption{Cumulative energy distribution and pointwise probability distance at $\beta=10^{-7}$, for the exact spectrum, 
the QPE-distorted spectrum (see Eq.~\eqref{eq:QPE-distorsion}), and for QMS measurements. 
Data has been obtained using $q_e = 7$ qubits for the energy register, with 1 rethermalization step (no plaquette measurement), $50$ thermalization steps,
about 55800 measurement samples with errors estimated via blocking and 100 bootstrap resamples.}
\label{fig:data_b0_r1_th50_ne7_KS}
\end{figure*}

\subsection{Thermal averages at finite temperatures}\label{subsec:numres-betafin}

As done in the previous section for vanishing values of $\beta$, 
here we discuss results at $\beta=0.1$ and $\beta=0.5$.
Fig.~\ref{fig:data_b0.1_r1_th50} and~\ref{fig:data_b0.5_r1_th50}
show the energy distributions in these two cases, 
while Fig.~\ref{fig:EneMeas_allb} reports 
the thermal averages of the energy estimated for all $\beta$ and numbers of qubits considered.
\begin{figure*}
  \includegraphics[width=0.95\textwidth]{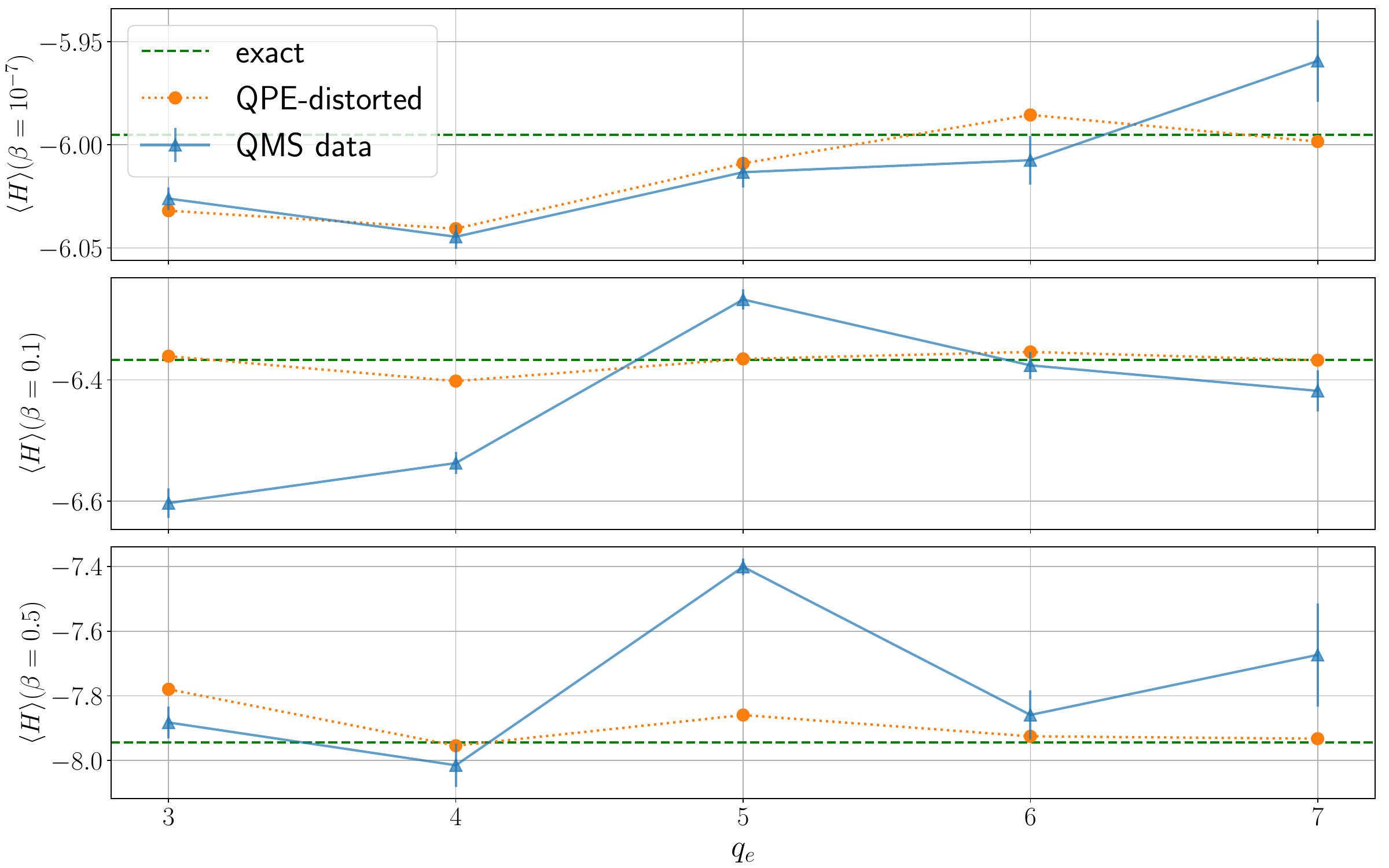}
    \caption{Thermal averages of the energy measured via QMS, and comparison with expected QPE-distorted estimate (see Eq.~\eqref{eq:QPE-distorsionExpval}) and exact value.}
    \label{fig:EneMeas_allb}
\end{figure*}

While the distribution at $\beta=0.1$ is relatively similar to the case of vanishing $\beta$, with the expected effect of the QPE distortion described in Sec.~\ref{subsec:exactqped} matching quite well with QMS measurements,
the behavior of data at $\beta=0.5$ appears worse.
%
\begin{figure*}
\begin{subfigure}{0.475\linewidth}
  \includegraphics[width=0.95\textwidth]{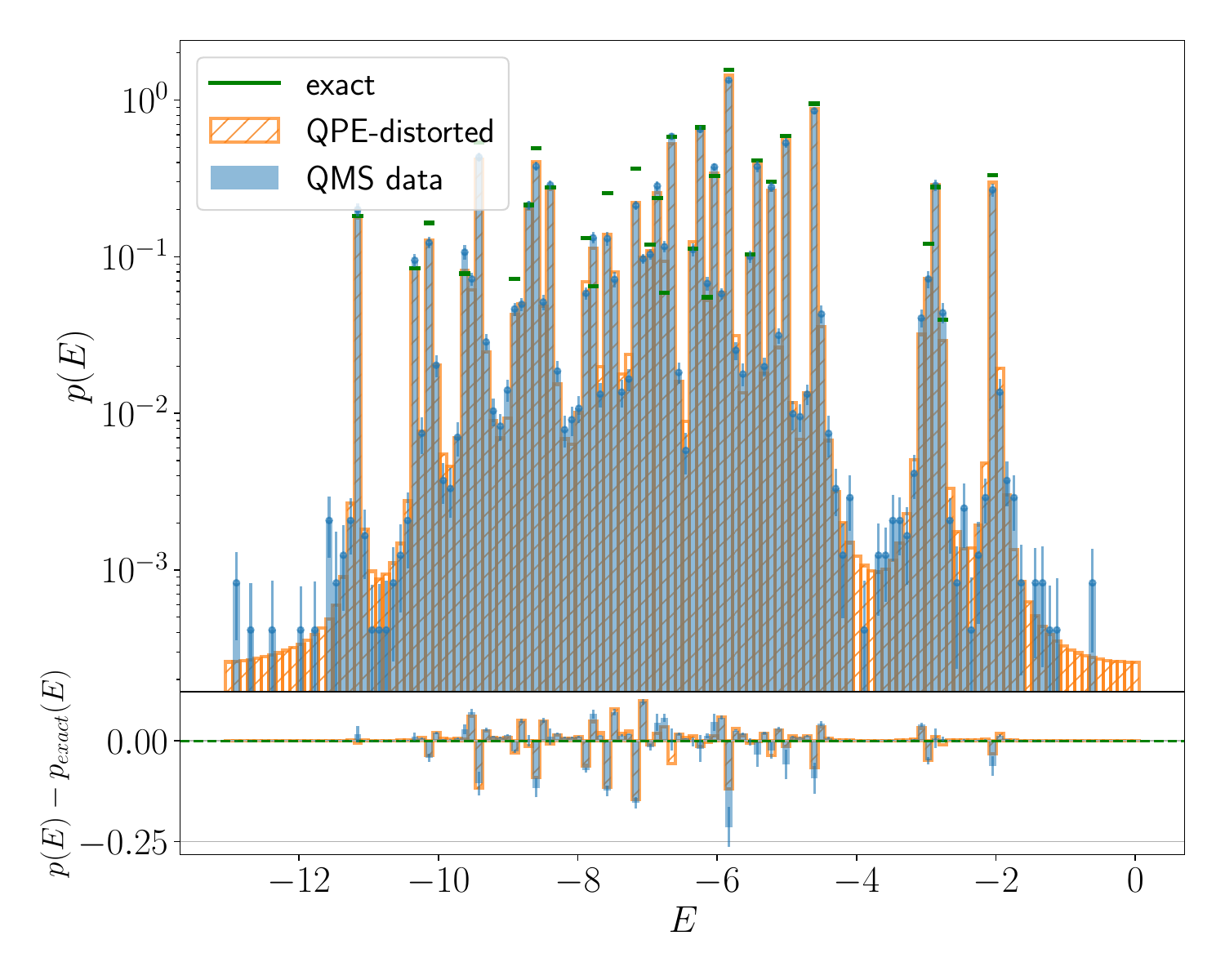}
  \caption{Binned histogram of energy distribution.}
  \label{fig:histo_b0.1_r1_th50}
\end{subfigure}\hfill
\begin{subfigure}{0.475\linewidth}
  \includegraphics[width=0.95\textwidth]{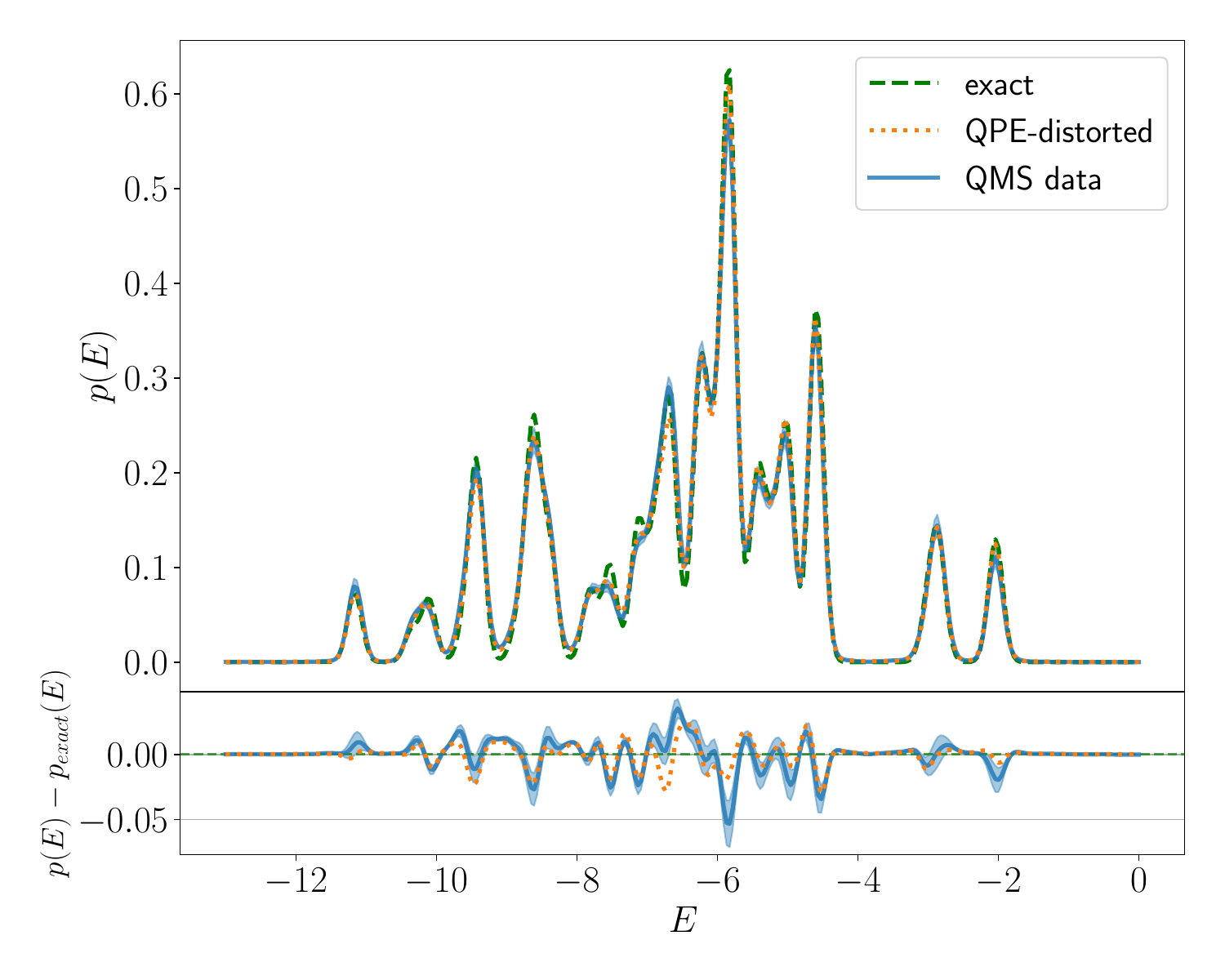}
  \caption{KDE of energy distribution ($\sigma_{\text{KDE}} \simeq 0.10$).}
  \label{fig:kde_b0.1_r1_th50}
\end{subfigure}
\medskip
\begin{subfigure}{0.475\linewidth}
  \includegraphics[width=0.95\textwidth]{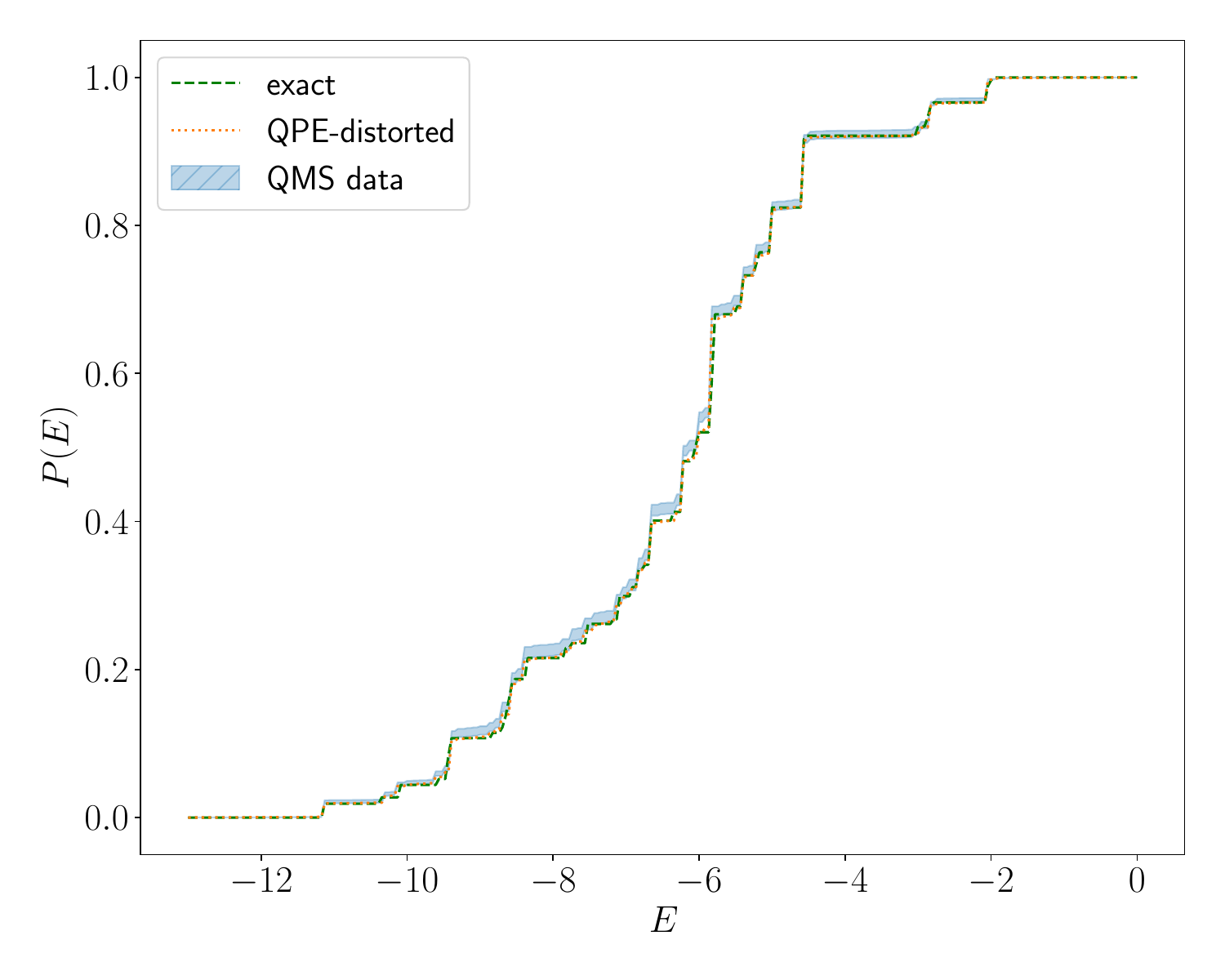}
\caption{Cumulative energy distributions.}
\label{fig:KS_distrib_b0.1_r1_th50}
\end{subfigure}
\begin{subfigure}{0.475\linewidth}
  \includegraphics[width=0.95\textwidth]{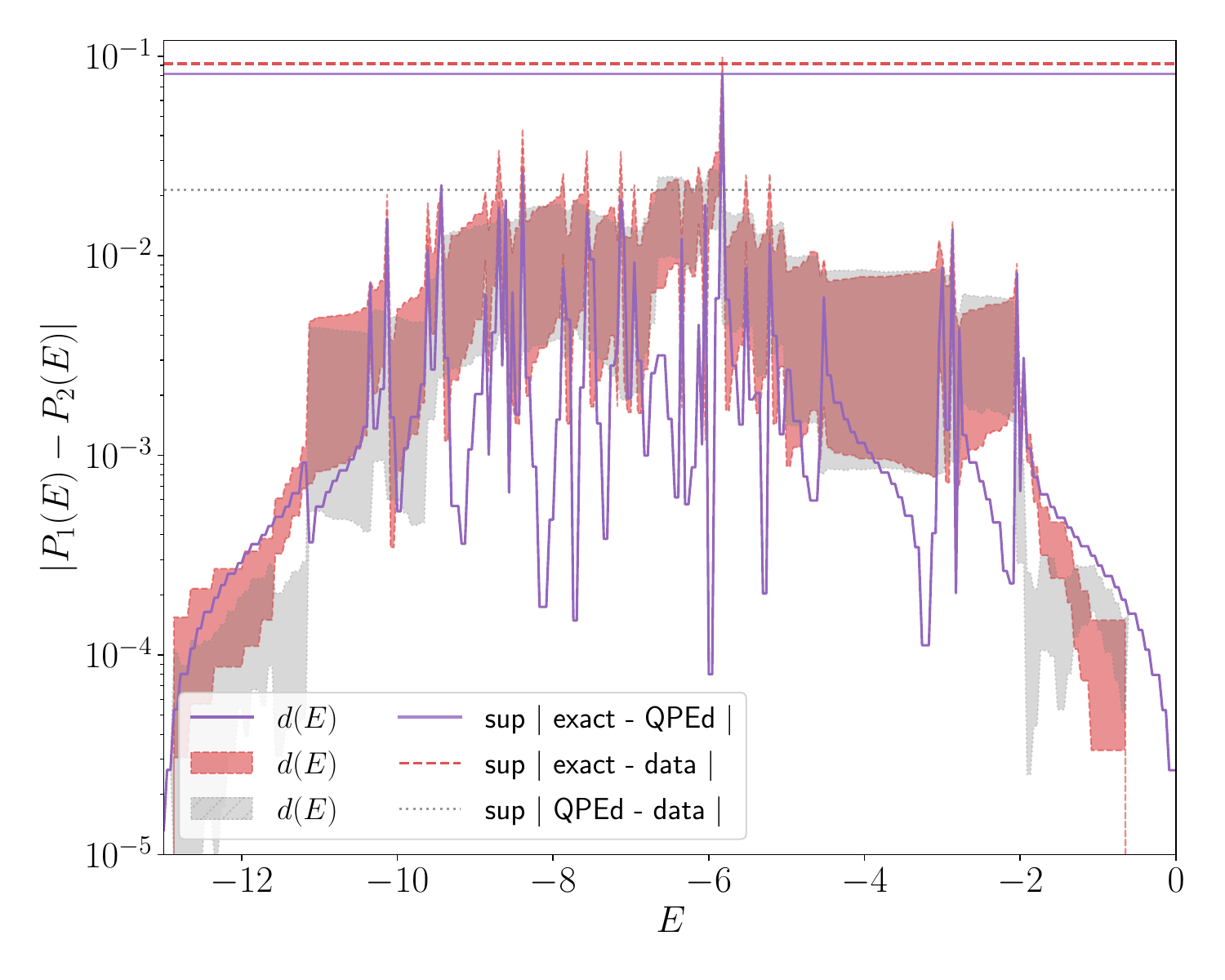}
\caption{Pointwise distance of cumulative energy distributions.}
\label{fig:KS_profile_b0.1_r1_th50}
\end{subfigure}\hfill
\caption{Energy distribution at $\beta=0.1$, for the exact spectrum, the QPE-distorted spectrum (see Eq.~\eqref{eq:QPE-distorsion}), and for QMS measurements. 
Data has been obtained using $q_e = 7$ qubits for the energy register, with 1 rethermalization step (no plaquette measurement), $50$ thermalization steps,
about 23600 measurement samples with errors estimated via blocking and 100 bootstrap resamples.}
\label{fig:data_b0.1_r1_th50}
\end{figure*}
%
\begin{figure*}
\begin{subfigure}{0.475\linewidth}
  \includegraphics[width=0.95\textwidth]{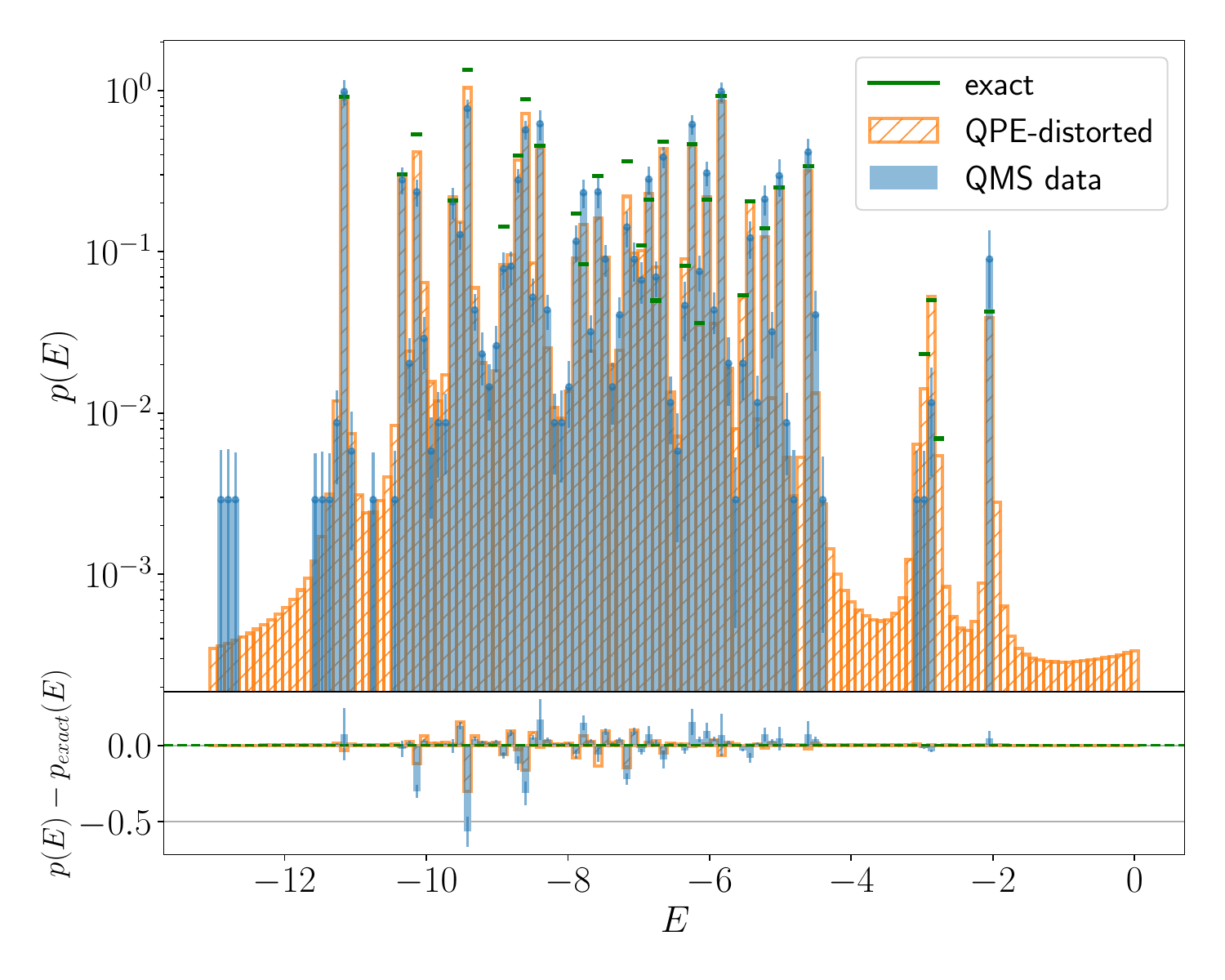}
  \caption{Binned histogram of energy distribution.}
  \label{fig:histo_b0.5_r1_th50}
\end{subfigure}\hfill
\begin{subfigure}{0.475\linewidth}
  \includegraphics[width=0.95\textwidth]{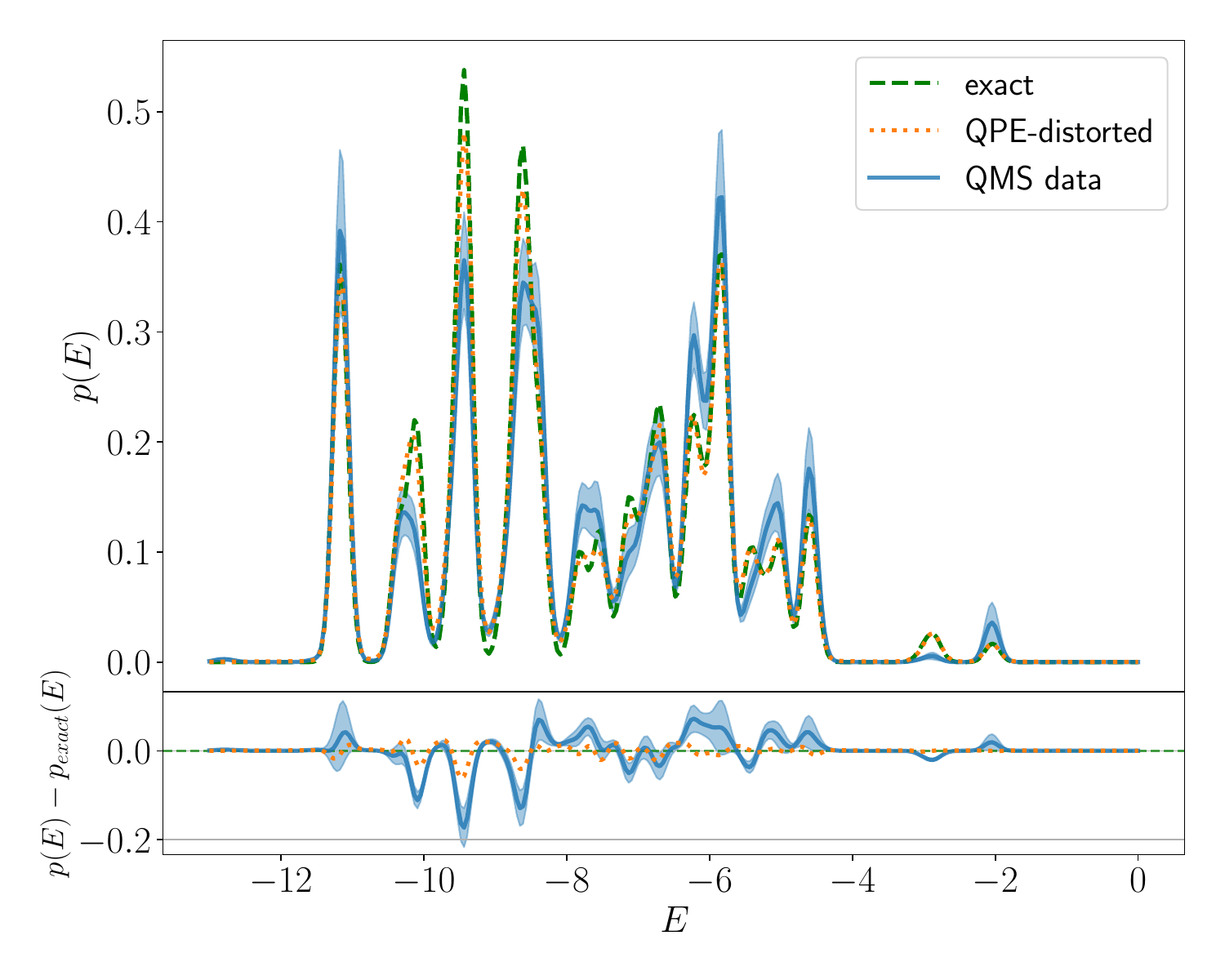}
  \caption{KDE of energy distribution ($\sigma_{\text{KDE}} \simeq 0.10$).}
  \label{fig:kde_b0.5_r1_th50}
\end{subfigure}
\medskip
\begin{subfigure}{0.475\linewidth}
  \includegraphics[width=0.95\textwidth]{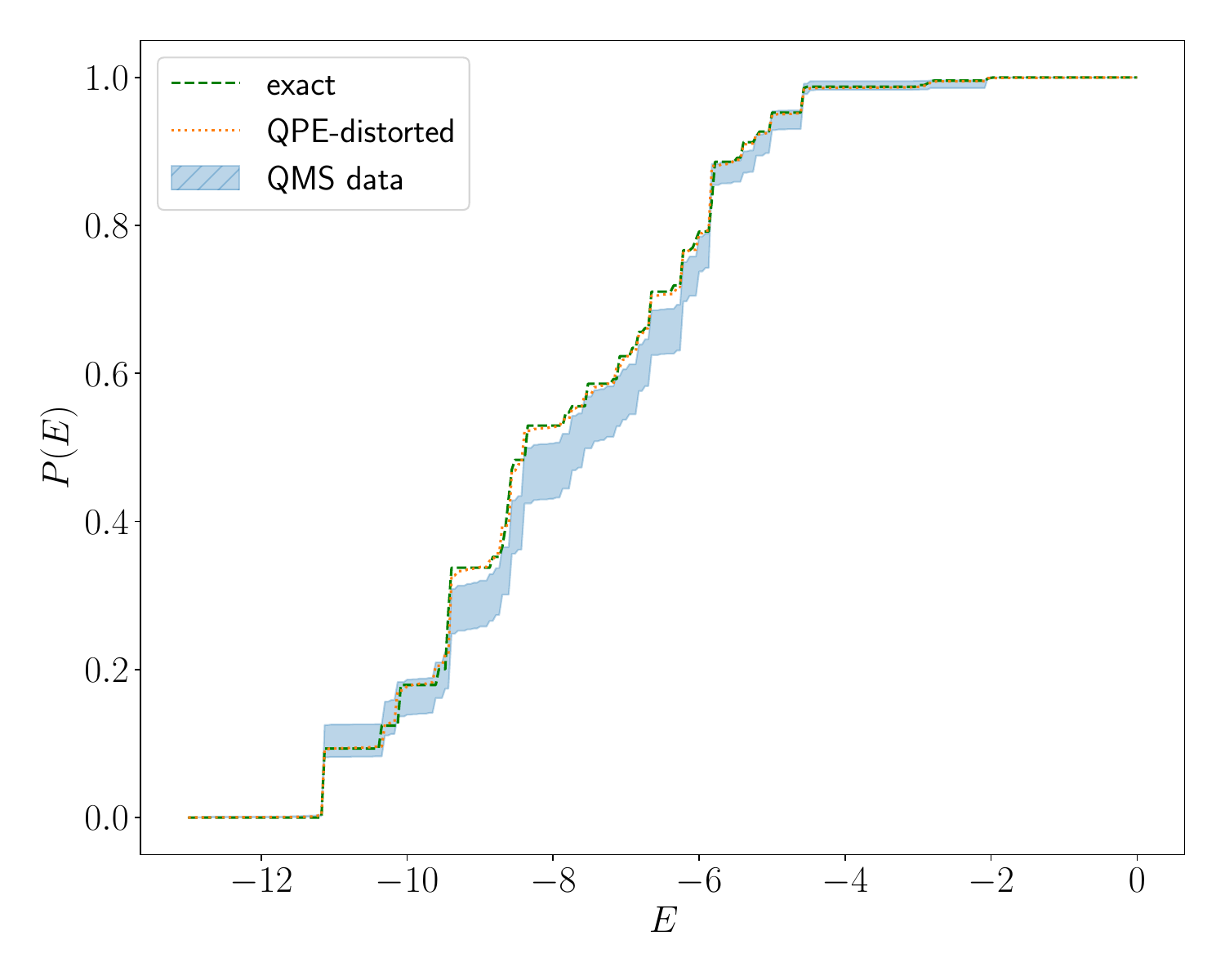}
\caption{Cumulative energy distributions.}
\label{fig:KS_distrib_b0.5_r1_th50}
\end{subfigure}
\begin{subfigure}{0.475\linewidth}
  \includegraphics[width=0.95\textwidth]{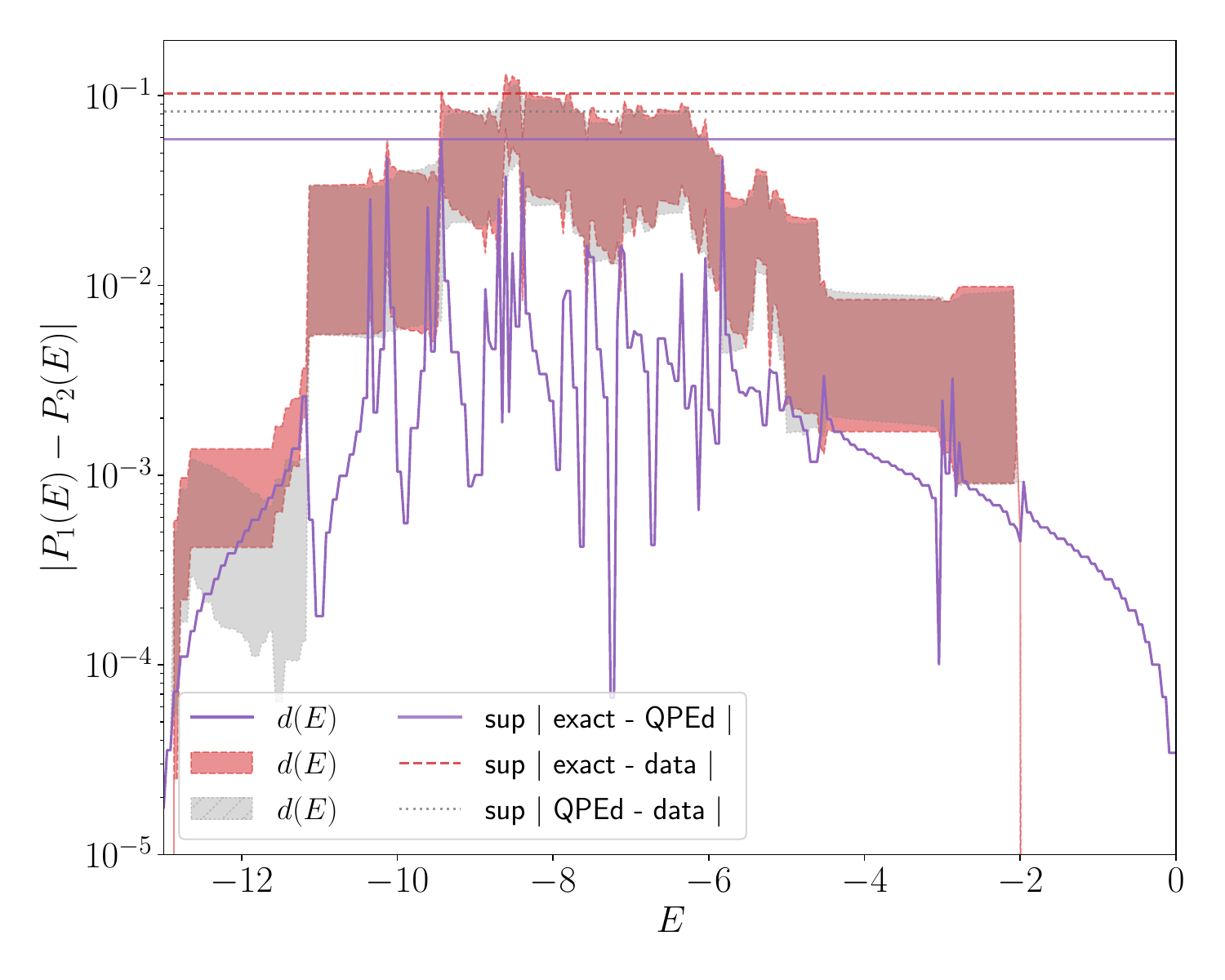}
\caption{Pointwise distance of cumulative energy distributions.}
\label{fig:KS_profile_b0.5_r1_th50}
\end{subfigure}\hfill
\caption{Energy distribution at $\beta=0.5$, for the exact spectrum, the QPE-distorted spectrum (see Eq.~\eqref{eq:QPE-distorsion}), and for QMS measurements. 
Data has been obtained using $q_e = 7$ qubits for the energy register, with 1 rethermalization step (no plaquette measurement), $50$ thermalization steps,
about 3400 measurement samples with errors estimated via blocking and 100 bootstrap resamples.}
\label{fig:data_b0.5_r1_th50}
\end{figure*}
\begin{figure*}
  \includegraphics[width=0.95\textwidth]{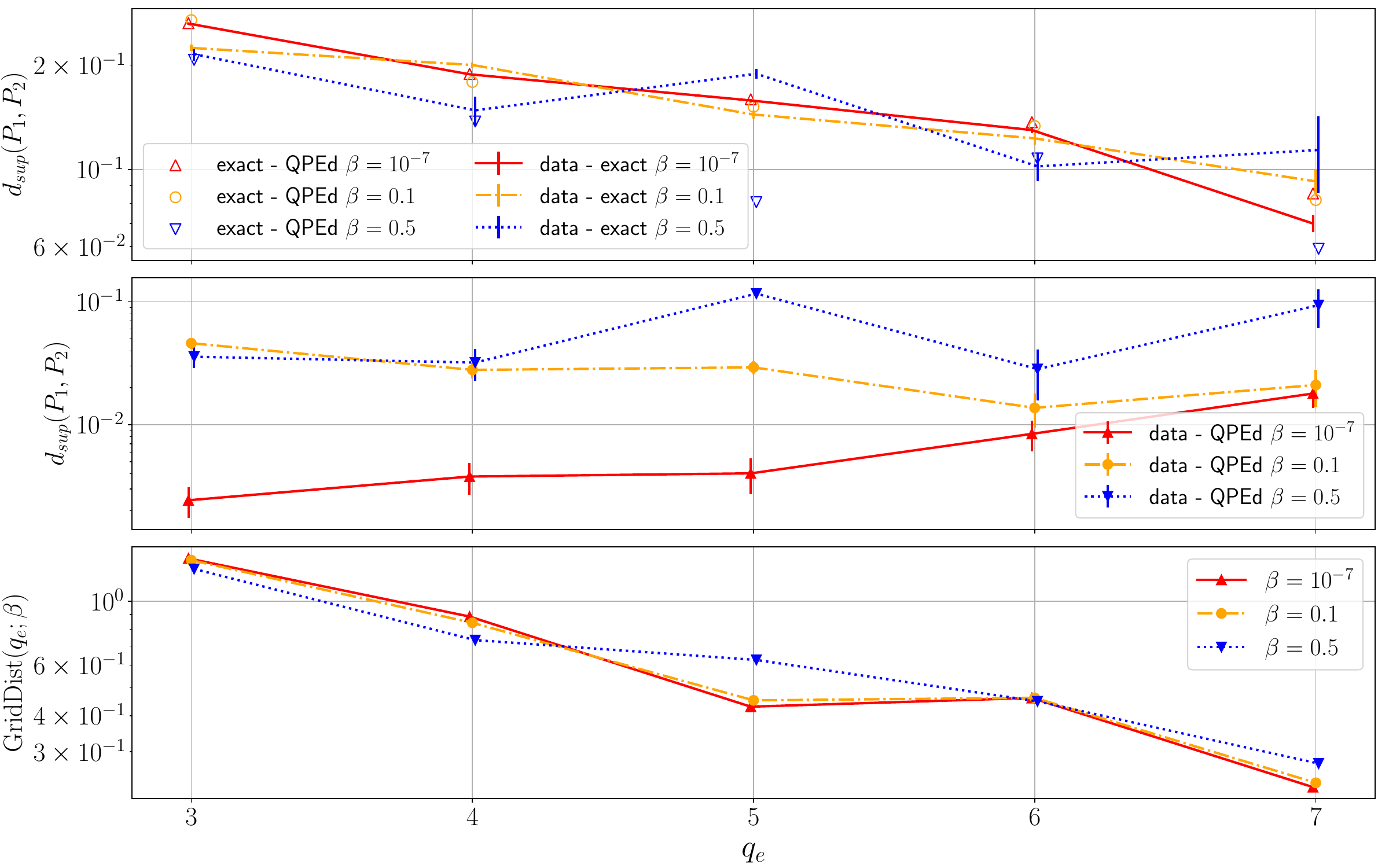}
\caption{Accuracy of distributions as maximum of the pointwise distance between cumulative distributions for the exact spectrum, the QPE-distorted spectrum and QMS data (top and middle panels), 
and $\mathrm{GridDist}$ estimate (defined in Eq.~\eqref{eq:GridDist_def}) between exact spectrum and QPE grid (bottom panel) as a function of the number of qubits $q_e$ in the energy register (with fixed measurement range $[-13,0]$) and for different values of inverse temperature $\beta$.
A small offset in the horizontal axis for different values of $\beta$ has been added to improve readability.}
\label{fig:KS_dist_r1_th50_GridDist}
\end{figure*}
The top and middle panels of Figure~\ref{fig:KS_dist_r1_th50_GridDist} 
show the behavior of the quantity
\begin{align}\label{eq:distsup_distr}
    d_{sup}(P_1,P_2) \equiv \sup_E \lvert P_1(E)- P_2(E)\rvert,
\end{align}
with $P_1$ and $P_2$ two distinct cumulative distributions.
While Fig.~\ref{fig:KS_profile_b0_r1_th50}, Fig.~\ref{fig:KS_profile_b0.1_r1_th50}, and Fig.~\ref{fig:KS_profile_b0.5_r1_th50}
give a detailed view of the relative discrepancies between the three distributions, the distance expressed in Eq.~\eqref{eq:distsup_distr} (which is the same used in Kolmogorov--Smirnov tests) puts a strong bound on the quality of convergence between the three kinds of distributions we consider, namely the one of the exact (physical) spectrum, of its QPE-distorted counterpart, and of the QMS energy measurements. 
The general behavior in the top and middle panels of Figure~\ref{fig:KS_dist_r1_th50_GridDist} seems to be consistent with the expected QPE distortion, 
discussed in Sec.~\ref{subsec:exactqped}, 
with a systematic error of QMS data which \textit{tends} to decrease by increasing the energy resolution (i.e., $q_e$),
even if not always in a monotonical fashion.
A particular exception to this is observed for the point at $\beta=0.5$ and $q_e=5$. Our tentative explanation is the following.
With fixed extrema of the grid and $\beta$, incrementing the number of grid sites might allow some of the eigenvalues of the real spectrum to be located in the middle between two sites of the grid, leaking a contribution to both, and making a measurement on one of the two neighboring grid points to set the state in a superposition of eigenstates with relatively similar energy but not exactly in the same eigenspace. The argument made in Sec.~\ref{subsec:exactqped} 
to account for the QPE distortion, which assumes an energy distribution coming from a stationary chain with distorted weights but still made of exact eigenstates, 
does not hold anymore, and this explains why data in this particular case do not follow the 
expected QPE-distorted distribution as well as 
for other numbers of qubits $q_e$ for the energy register. 
This effect is similar to what is experienced with floating point artifacts in classical computing, and it is expected to vanish (in general, non-monotonically) by increasing $q_e$. 
At the same time, even using the fixed grid with $q_e=5$, 
this effect shows up in particular when increasing $\beta$ from $0.1$ to $0.5$.
Therefore, the weights $p_k^{(QPEd)}$ used in Eq.~\eqref{eq:QPE-distorsion} 
to predict the effects of the QPE distortion are not accurate enough.
In order to get an intuitive picture of the reasons for this effect,
we introduce a quantity that assesses the average weighted distance between
the QPE grid and the spectrum.
Let us consider a real spectrum $\Sigma=\{\sigma_k\}$, with weights $p_k$ and a uniform grid \mbox{${\{x_j=a+\frac{(b-a)}{(2^{q}-1)} j\}}_{j=0}^{q}$} in the range $[a,b]$. 
We can define a measure of the relative QPE-weighted distance between the spectrum and the grid as the following quantity:
\begin{align}
    \mathrm{GridDist}&(q;\beta)
    \equiv \sum_{\sigma_k\in \Sigma} p_k(\beta) \sqrt{\sum\limits_{j=0}^{2^q-1} {\lvert \sigma_k - x_j \rvert}^2 \cdot {|c_{k,j}|}^2},
    \label{eq:GridDist_def}
\end{align}
where the weights ${|c_{k,j}|}^2$ depend on both real spectrum and grid according to Eq.~\eqref{eq:QPEcoeffs}.
The bottom panel in Figure~\ref{fig:KS_dist_r1_th50_GridDist} shows how the particular 
anomalous point discussed above is affected by a higher relative 
weighted distance than the other points, 
which would generally be expected to follow 
a monotonically decreasing trend. 
Since the quantity $\mathrm{GridDist}$ correlates well with 
the behavior observed in QMS data 
(top panel in Figure~\ref{fig:KS_dist_r1_th50_GridDist}), 
it appears reasonable to assume that the source of QPE distortion 
is not completely predicted by Eq.~\eqref{eq:QPE-distorsion}, 
which represents only a rough approximation for $\beta > 0$.
Furthermore, this effect is amplified by the fact that the spectrum 
studied here is not quasi-continuous, 
due to the small volume of the lattice and the 
finiteness of the gauge group.

\subsection{Gauge invariant measurement}\label{subsec:numres-plaqmeas}
As discussed in Sec.~\ref{subsec:retherm-GImeas}, in order to perform a rethermalization step instead of thermalization, it is necessary to keep the state gauge invariant also during the measurement of an observable not commuting with the Hamiltonian. 
In the case of the thermal average of the (trace of the) left plaquette operator, a measurement with the protocol discussed in the caption to Fig.~\ref{fig:GImeasD4} yields only three possible values, namely $-2$, $0$, and $2$. 
Since both data and exact distribution lie on that discrete domain, 
it does not make sense to compare them using histograms or KDE. 
Instead, we report the distribution in Table~\ref{tab:plaqMeas_b0} for $\beta=10^{-7}$, Table~\ref{tab:plaqMeas_b0.1} for $\beta=0.1$ and Table~\ref{tab:plaqMeas_b0.5} for $\beta=0.5$.
\begin{table}[h]
    \centering
    \begin{tabular}{c | c c c | c}
        $q_e$ & $p_{-2}$ & $p_{0}$ & $p_{2}$ & $\langle{\Tr\mathcal{P}_{0y}}\rangle$ \\
       \hline
       $3$ & $0.159(2)$ & $0.685(3)$ & $0.157(2)$ & $-0.004(6)$\\
       $5$ & $0.159(1)$ & $0.680(2)$ & $0.160(1)$ & $-0.002(3)$\\
       $7$ & $0.160(3)$ & $0.681(4)$ & $0.159(3)$ & $-0.004(8)$\\
       \hline
       exact & $0.15909$ & $0.68182$ & $0.15909$ & $0.0$
    \end{tabular}
    \caption{Distribution of trace of plaquette at $\beta=10^{-7}$.
            Thermalization steps: 50, Rethermalization steps: 20, blocksize: 50;
            The number of samples collected for $q_e=3$, $q_e=5$, and $q_e=7$ are respectively about 38600, 122000, and 17900.}
    \label{tab:plaqMeas_b0}
\end{table}
\begin{table}[h]
    \centering
    \begin{tabular}{c | c c c | c}
        $q_e$ & $p_{-2}$ & $p_{0}$ & $p_{2}$ & $\langle{\Tr\mathcal{P}_{0y}}\rangle$ \\
       \hline
       $3$ & $0.132(4)$ & $0.670(6)$ & $0.199(5)$ & $0.133(15)$\\
       $5$ & $0.130(3)$ & $0.679(5)$ & $0.190(1)$ & $0.123(9)$\\
       $7$ & $0.131(8)$ & $0.676(15)$ & $0.193(10)$ & $0.124(28)$\\
       \hline
       exact & $0.12331$ & $0.67295$ & $0.20374$ & $0.0536$
    \end{tabular}
    \caption{Distribution of trace of plaquette at $\beta=0.1$.
            Thermalization steps: 50, Rethermalization steps: 20, blocksize: 50;
            The number of samples collected for $q_e=3$, $q_e=5$, and $q_e=7$ are respectively about 6200, 16600, and 2400.}
    \label{tab:plaqMeas_b0.1}
\end{table}
\begin{table}[h]
    \centering
    \begin{tabular}{c | c c c | c}
        $q_e$ & $p_{-2}$ & $p_{0}$ & $p_{2}$ & $\langle{\Tr\mathcal{P}_{0y}}\rangle$ \\
       \hline
       $3$ & $0.061(8)$ & $0.52(2)$ & $0.42(2)$ & $0.71(4)$\\
       $5$ & $0.075(6)$ & $0.59(1)$ & $0.34(1)$ & $0.53(3)$\\
       $7$ & $0.049(16)$ & $0.53(4)$ & $0.42(4)$ & $0.7(1)$\\
       \hline
       exact & $0.04349$ & $0.49712$ & $0.45940$ & $0.27727$
    \end{tabular}
    \caption{Distribution of trace of plaquette at $\beta=0.5$.
            Thermalization steps: 50, Rethermalization steps: 20, blocksize: 10;
            The number of samples collected for $q_e=3$, $q_e=5$, 
            and $q_e=7$ are respectively about 850, 1500, and 123.}
    \label{tab:plaqMeas_b0.5}
\end{table}
The presence of a finite resolution for energy measurements (i.e., the QPE grid),
discussed in Sec.~\ref{subsec:exactqped} and manifest 
in the results of Sec.~\ref{subsec:numres-betafin}, 
does not only affect the energy distribution measured with QMS, 
but reflects also on the measured distribution of other observables such as the average plaquette.
For this reason, the degradation in the quality of the results at higher values of $\beta$ is
all the more manifest in this case.
In general, increasing the number of qubits $q_e$ would improve the quality of the results,
even if not in a monotonic fashion (see discussion in Sec.~\ref{subsec:numres-betafin}), 
but the lower acceptance probability and the slower speed of emulation  
make it difficult to obtain a better estimate of these effects 
for an asymptotically large number of qubits.

\section{Conclusions}\label{sec:conclusions}
To summarize, this paper treats the problem of studying lattice gauge theories
with the Quantum Metropolis Sampling algorithm, 
discussing in particular the implementation for 
a (2+1)-dimensional lattice gauge theory with finite gauge group $D_4$. 
The challenges we encountered and tackled along the process include the following:
\begin{itemize}
    \item determining how to build a set of Metropolis quantum updates (i.e., unitary operators) 
    which preserves both gauge invariance of the states and is ergodic on the physical Hilbert 
    space of the system (in the sense discussed in Sec.~\ref{subsec:GImoves});
    \item building a protocol to perform measurements of physical observables without breaking the 
    gauge invariance of the state after measurement;
    \item taking into account the distortion introduced by a finite energy resolution 
    for Quantum Phase Estimation (QPE) to predict the expected energy distribution of the
    QMS measurements.
\end{itemize}
The numerical results for $\beta=0$, presented in Sec.~\ref{subsec:numres-beta0}, demonstrate that 
the conditions of ergodicity and gauge invariance are satisfied, 
and the general behavior, while presenting some discrepancy with the exact diagonalization 
results, show an excellent agreement with the expectations coming from taking into account
the effects of QPE distortion discussed in Sec.~\ref{subsec:exactqped}.
For higher values of $\beta$, the sampling becomes less efficient because of a generally smaller
acceptance rate, and the sampled distribution exhibits an even higher distortion as shown from
the results in Sec.~\ref{subsec:numres-betafin}. We provide a tentative explanation of
the reasons for these discrepancies by introducing a quantity ($\mathrm{GridDist}$) 
which quantifies heuristically the Gibbs-weighted average square distance between 
the QPE grid on which energy measurements are performed, and the actual spectrum of the system.
We show that this quantity correlates well with the source of distortion,
which, in general, does not follow a monotonically decreasing trend of the systematic error.
This effect, due to a mismatch between the real spectrum and the QPE grid, has been also observed in Ref.~\cite{Aiudi:2023cyq}.

As future perspectives, it would be interesting to investigate systematically how the choice of different sets of moves affects efficiency.
Furthermore, one can introduce fermionic matter or a topological theta term
on similar gauge systems with the application of other 
algorithms of thermal average estimation proposed in literature.
Indeed, since the gauge adaptations we made to the QMS algorithm 
require ideas applied to several quantum components 
besides the already well-known time evolution for $D_4$ 
(used to define the Quantum Phase Estimation, Sec.~\ref{subsec:exactqped}),
we believe that the results and ideas we considered in this work might be useful 
also for different quantum algorithms meant to be applied to general lattice gauge theories.

\acknowledgments
We thank Claudio Bonati and Phillip Hauke for useful discussions.
GC thanks the INFN of Pisa for the hospitality while writing this manuscript.
EB thanks funding by the European Union under Horizon Europe Programme – Grant Agreement 101080086 – NeQST. 
KZ thanks funding by the University of Pisa under the “PRA - Progetti di Ricerca di Ateneo” (Institutional Research Grants) - Project No. PRA 2020-2021 92 “Quantum Computing, Technologies and Applications”.
MD acknowledges support from the National Centre on HPC, Big Data and Quantum Computing - SPOKE 10 (Quantum Computing) and received funding from the European Union Next-GenerationEU - National Recovery and Resilience Plan (NRRP) – MISSION 4 COMPONENT 2, INVESTMENT N. 1.4 – CUP N. I53C22000690001.
Views and opinions expressed are however those
of the author(s) only and do not necessarily reflect those of the European Union or European Climate, Infrastructure and Environment Executive Agency (CINEA). 
Neither the European Union nor the granting authority can be held responsible for them.
This project has received funding from the European Research Council (ERC) under the European Union’s Horizon 2020 research and innovation programme (grant agreement No 804305).
Numerical simulations have been performed on the \texttt{Marconi100} machines at CINECA, based on 
the agreement between INFN and CINECA, under project INF23\_npqcd.

\appendix

\section{Overview of lattice gauge theories with finite groups and matching between transfer matrix formulations}\label{app:transfmat_gropth}
Here we briefly review some concepts of gauge theories with finite gauge groups.
In particular,  we establish the precise connection between the Hamiltonian and transfer matrix
formulation by discussing the match between Casimir coefficients for the
kinetic part of the Hamiltonian (which can be obtained also from group-theoretical considerations, i.e., see~\cite{Mariani:2023eix}),
and the transfer matrix approach based on the Euclidean Lagrangian formulation.
The relation between real space (link basis) and irrep basis is the following~\cite{Mariani:2023eix}:
\begin{align}
   \braket{g}{\widetilde{j\alpha \beta}} = \sqrt{\frac{d_j}{|G|}} {\rho_j(g)}^{\alpha}_\beta, 
\end{align}
where $j$ labels the unitary irreps $\rho_j: G \to U(d_j)$ with dimension $d_j$,
and we denote by $\ket{\widetilde{j \alpha \beta}}$ the irreps basis to distinguish it from the real space basis.
Notice that there are two indexes for each irrep $\alpha$ $\beta$, labeling the matrix element, and that $\sum_j d_j^2 = |G|$ (a sum over $j$ will always denote a summation over all irreps). 
For the \emph{Peter--Weyl theorem}, the irreps form then a complete basis that can be considered as an analog of the Fourier basis. Let us introduce also the projectors into irrep spaces:
\begin{align}\label{eq:irrepProj}
    \mathbb{P}_j \equiv \sum_{\alpha \beta} \ketbra{\widetilde{j \alpha \beta}}{\widetilde{j \alpha \beta}},
\end{align}
which, by Peter-Weyl and the completeness relation for irreps, 
form a partition of identity $\sum_j \mathbb{P}_j=\matone$ (and also, satisfy the relations $\mathbb{P}_j \mathbb{P}_{j^\prime}=\delta_{j j^\prime} \mathbb{P}_j$).

Using $\hat{L}_{[g]}\equiv \sum_{\tilde{g}} \ketbra{\tilde{g}}{g \tilde{g}}$ and $\hat{R}_{[g]}\equiv \sum_{\tilde{g}} \ketbra{\tilde{g}}{\tilde{g} g^\dagger}$,
the form of the projectors in real space is the following:
\begin{align}
    \mathbb{P}_j &= \sum_{\alpha, \beta, g, g^\prime} \frac{d_j}{|G|} {\rho_j(g^\prime)}_{\alpha \beta}{{\rho_j(g)}^\dagger}_{\beta \alpha} \ketbra{g^\prime}{g} 
    =\frac{d_j}{|G|} \sum_h \chi_j(h) \hat{R}^\dagger_{[h]},
\end{align}
where $\chi_j$ is the character of the $j$-th irrep.
In terms of irreps, the Right and Left multiplication operators take the form
\begin{align}
    \hat{L}_{[g]} &= \sum_{g^\prime} \ketbra{g^\prime}{g g^\prime}
    =\sum_{j,\alpha^\prime,\alpha,\beta} {[\rho_{j}(g)]}_{\alpha,\alpha^\prime} \ketbra{\widetilde{j \alpha^\prime \beta}}{\widetilde{j \alpha \beta}},\\
    \hat{R}_{[g]} &= \sum_{g^\prime} \ketbra{g^\prime}{g^\prime g^\dagger}
    =\sum_{j,\alpha,\beta^\prime,\beta} {[\rho_{j}(g)]}_{\beta,\beta^\prime} \ketbra{\widetilde{j \alpha \beta^\prime}}{\widetilde{j \alpha \beta}},
\end{align}
where we used the orthogonality relations: 
\begin{align}
    \sum_g {[\rho_{j^\prime}(g)]}_{\alpha^\prime,\beta^\prime} {[{\rho_{j}(g)}^*]}_{\alpha,\beta}=\frac{|G|}{d_j}\delta_{j^\prime,j}\delta_{\alpha^\prime, \alpha}\delta_{\beta^\prime \beta}.
\end{align}
With some more effort, also the link operator (non-Hermitian) can be written in
terms of the irreps.
Indeed, using the Clebsch-Gordan coefficients for the group, 
which are defined via the relation
\begin{align}\label{eq:D4CGcoeffs}
    {\rho_{j^\prime}(g)}_{\alpha^\prime, \beta^\prime}^*
    {\rho_{j^{\prime\prime}}(g)}_{\alpha^{\prime\prime}, \beta^{\prime\prime}} = 
    \sum_{J,A,B} C^{(J,A,B)}_{(j^\prime, \alpha^\prime, \beta^\prime); (j^{\prime\prime}, \alpha^{\prime\prime}, \beta^{\prime\prime})} {\rho_{J}(g)}_{A,B}^*.
\end{align}
one can write
\begin{equation}
\begin{aligned}
    \hat{U}_{\alpha,\beta} =\! \sum_{g\in G} {\rho_f(g)}_{\alpha, \beta} \ketbra{g}{g}
    \!=\!\! \sum\limits_{\vec{J}', \vec{J}''}\!\!\!\!
    \frac{\sqrt{d_{J'}d_{J''}}}{d_f}
    C^{(f,\alpha,\beta)}_{\vec{J}';\vec{J}''} \ketbra{\widetilde{\vec{J}'}}{\widetilde{\vec{J}''}},
\end{aligned}
\end{equation}
where we use the shorthand $\vec{J}'\equiv (j',\alpha',\beta')$ and $\vec{J}''\equiv (j'',\alpha'',\beta'')$,
which identifies different matrix elements for each irrep with a single multi-index.

For a projector $\mathbb{P}$, we have $e^{\alpha \mathbb{P}} = \matone + (e^\alpha -1)\mathbb{P}= (\matone-\mathbb{P}) + e^\alpha \mathbb{P}$, therefore, for a partition of unity set of projectors 
$\{\mathbb{P}_j\}_{j}$, such that $\oplus_j \mathbb{P}_j = \matone$ and 
$\mathbb{P}_j \mathbb{P}_{j^\prime} = \mathbb{0}$, there is a group-theoretical (GT)
motivated expression for the electric/kinetic term of the Hamiltonian based on the Casimir operator as a sum over irreps and constant inside conjugacy classes:
\begin{align}\label{eq:appA-HGT}
    \hat{H}_k^{(GT)} = \alpha \sum_j f_j \mathbb{P}_j. 
\end{align}
Indeed, as discussed in~\cite{Mariani:2023eix}, 
this is the most general Hamiltonian that can be used 
to describe a gauge invariant theory.
The transfer matrix corresponding to a finite (Euclidean) time $\Delta t=1$ integration with the Hamiltonian~\eqref{eq:appA-HGT} becomes then:
\begin{align}\label{eq:appA-Tk_GT}
    \hat{T}_k^{(GT)}&\equiv e^{-\alpha \sum_j f_j \mathbb{P}_j} 
    =\sum_j e^{-\alpha f_j} \mathbb{P}_j\\
    &
    =  \sum_h \Big[\sum_j\frac{d_j}{|G|} e^{-\alpha f_j} \chi_j(h)\Big] \hat{R}^\dagger_{[h]} .
\end{align}

The (Euclidean) Lagrangian formulation of the transfer matrix~\cite{Menotti:1981ry,Lamm:2019bik} instead reads
{\small
\begin{align}\label{eq:appA-Tk_L}
    \hat{T}_k^{(L)}&\equiv \sum_{g^\prime,g} e^{\beta \Re\Tr[\rho_f(g^\prime g^\dagger)]} \ketbra{g^\prime}{g} 
    = \sum_{h} e^{\frac{\beta}{2} (\chi_f(h)+\chi_f(h^{-1}))}\hat{R}_{[h]},
\end{align}
}
where $\rho_f$ is a fundamental representation and $\chi_f$ is the corresponding character and $\beta = \frac{1}{g^2}$.
Matching the two expressions~\eqref{eq:appA-Tk_L} and~\eqref{eq:appA-Tk_GT} can be done by noticing the same structure:
\begin{align}\label{eq:matchingGT_L}
   \frac{1}{|G|} \sum_j d_j e^{-\alpha f_j} \chi_j(h) \longleftrightarrow e^{\frac{\beta}{2} (\chi_f(h)+\chi_f(h^{-1}))}\qquad \forall{h\in G},
\end{align}
which can be inverted in terms of $\alpha$ and the coefficients $f_j$ 
using character orthogonality, such that
\mbox{($\frac{1}{|G|}\sum_g \chi_j(g) \chi_{j^\prime}(g) = \delta_{j, j^\prime}$)}
\begin{align}\label{eq:matchingGT_L_inverse}
   e^{-\alpha f_j} = \frac{1}{d_j} \sum_{g\in G} e^{\frac{\beta}{2} (\chi_f(g)+\chi_f(g^{-1}))} \chi_j(g) \qquad \forall j.
\end{align}
In the case of continuous groups, Eq.~\eqref{eq:matchingGT_L_inverse}
can be further manipulated with a saddle point expansion for $\beta \to \infty$
(see for example Ref.~\cite{Creutz:1976ch}),
resulting in a match of the type $\alpha f_j \mathbb{P}_j \sim g^2 \sum_a {(\hat{E}_j^{a})}^2$,
with $\hat{E}_j^{a}$ being the components of the electric field operators associated to the $j$-th irrep. For finite groups, an expansion in terms of $g^2$ is not available, as argued in the following Section, but one can still determine the dominant term in the strong coupling regime.

\subsection{Non-Abelian finite group: $D_4$}\label{subapp:transfmatD4}
For $G=D_4$, there are $5$ conjugacy classes $C_0=\{e\}$, $C_1=\{r,r^3\}$, $C_2=\{r^2\}$, $C_3=\{s,sr^2\}$, $C_4=\{sr,sr^3\}$, and $5$ irreps with characters as shown in Table~\ref{tab:D4irreps}.
\begin{table}[h]
    \centering
    \begin{tabular}{c|c|ccccc}
       $j$  & $d_j$ & $\chi_j \{e\}$ & $\chi_j \{r,r^3\}$ & $\chi_j \{r^2\}$ & $\chi_j \{s,sr^2\}$ & $\chi_j \{sr,sr^3\}$ \\
       \hline
       $0$ & 1 & 1& 1& 1& 1& 1 \\
       $1$ & 1 & 1& 1& 1& -1& -1 \\
       $2$ & 1 & 1& -1& 1& 1& -1 \\
       $3$ & 1 & 1& -1& 1& -1& 1 \\
       $4$ & 2 & 2& 0& -2& 0& 0\\
    \end{tabular}
    \caption{Characters of irreducible representations $\rho_j$ of $D_4$ group}
    \label{tab:D4irreps}
\end{table}
The fundamental representation is $j=4$, the only one with dimension 2, and can 
be expressed as a real representation $\rho_f(g)=\rho_f(s^{x_2} r^{2x_1+x_0}) = {(\sigma^x)}^{x_2} {(i\sigma^y)}^{2 x_1 + x_0}$ 
parameterized by a triple of binary digits $(x_2,x_1,x_0)\in \mathbb{Z}_2^3$,
while the other $1$-dimensional irreps all coincide with their characters.
The Casimir eigenvalues $f_j$ can be computed using the matching conditions in Eq.~\eqref{eq:matchingGT_L_inverse}, which results in
\begin{align}
    f_0 &= -\frac{1}{\alpha}\log (6+2\cosh(2\beta)),\\
    f_1=f_2=f_3 &= -\frac{1}{\alpha}\log (4\sinh^2(\beta)),\\
    f_4 &= -\frac{1}{\alpha}\log(2 \sinh(2\beta)).
\end{align}
Notice that the irrep labeled with $j=4$ corresponds to the fundamental representation (being $d_{j\leq 3}=1$), and contributes to the projector in Eq.~\eqref{eq:irrepProj} with $d_4^2=4$ states (one for each matrix element).
At this point, one can proceed in conventionally fixing the Casimir eigenvalue of the zero-mode 
(i.e., the trivial irrep $\rho_0$) to zero,
which means that shifting all the coefficients 
(in other terms, isolating the overall scalar prefactor of the transfer matrix) results in
\begin{align}\label{eq:D4fcoeffs_asympt}
    \tilde{f}_0 = 0,\quad \tilde{f}_1=\tilde{f}_2=\tilde{f}_3 \sim \frac{8 e^{-2\beta}}{\alpha},\quad \tilde{f}_4 \sim \frac{6 e^{-2\beta}}{\alpha}.
\end{align}
In the case of continuous groups, 
one can match in the usual saddle point expansion for small $\frac{1}{\beta}=g^2$
(see Ref.~\cite{Creutz:1976ch}) but, 
for finite groups, this is not available. 
Instead, one can express the Hamiltonian in terms of $e^{-2\beta}$ for large $\beta$ (or equivalently, small $g$).
Notice that, with the identification $\alpha=e^{-2\beta}$, one recovers the same Casimir coefficients 
found in Ref.~\cite{Mariani:2023eix}, using the elements in the conjugacy classes $\Gamma=C_1 \cup C_3 \cup C_4 = \{r,r^3,s,sr,sr^2,sr^3\}$ as generators,
which yields $f_{j=0}=0$, $f_{j=1,2,3}=8$, and $f_{j=4}=6$.
Therefore, following Eq.~\eqref{eq:D4fcoeffs_asympt}, 
a possible group-theoretical Hamiltonian expression for $D_4$,
compatible with the transfer matrix integrated with a fixed finite time as the one used in
Ref.~\cite{Lamm:2019bik}, is the following
{\footnotesize
\begin{align}\nonumber
    \hat{H}^{(GT,D_4)} &= \gamma e^{-2\beta}\bigoplus_{l\in E} \Big[ 8 \sum\limits_{j=1}^{3}{\Big(\ketbra{\widetilde{j}}{\widetilde{j}}\Big)}_l+6\sum_{\alpha, \beta} {\Big(\ketbra{\widetilde{4 \alpha \beta}}{\widetilde{4 \alpha \beta}}\Big)}_l\Big]\\
    &-\beta \bigoplus_{p\in Plaq.} \Re\Tr  \square_p,
\end{align}
}
where states corresponding to 1-dimensional irreps are denoted by a single number $\ket{\widetilde{j}}$ instead of a triple $\ket{\widetilde{j,\alpha,\beta}}$.
In cases when the model is expected to represent continuum physics, or even to build effective theories, the anisotropy coefficient $\gamma$ still needs to be tuned in order to preserve the lines of constant physics.
For completeness, we also report the non-vanishing Clebsch-Gordan coefficients for $D_4$, 
shown in Table~\ref{tab:CG_D4}.
\begin{table}[h]
    \centering
    \begin{tabular}{c|c|c|c}
       $(J,A,B)$  & $(j^\prime, \alpha^\prime, \beta^\prime)$ & $(j^{\prime\prime}, \alpha^{\prime\prime}, \beta^{\prime\prime})$ & $C^{(J,A,B)}_{(j^\prime,\alpha^\prime,\beta^\prime);(j^\prime,\alpha^\prime,\beta^\prime)}$ \\
       \hline
       (0) & (0) & (0) & 1 \\
       (0) & (1) & (1) & 1 \\
       (0) & (2) & (2) & 1 \\
       (0) & (3) & (3) & 1 \\
       (0) & (4,0,0) & (4,0,0) & 1/2 \\
       (0) & (4,0,1) & (4,0,1) & 1/2 \\
       (0) & (4,1,0) & (4,1,0) & 1/2 \\
       (0) & (4,1,1) & (4,1,1) & 1/2 \\
       (1) & (0) & (1) & 1 \\
       (1) & (2) & (3) & 1 \\
       (1) & (4,0,0) & (4,1,1) & 1/2 \\
       (1) & (4,0,1) & (4,1,0) & -1/2 \\
       (1) & (4,1,0) & (4,0,1) & -1/2 \\
       (1) & (4,1,1) & (4,0,0) & 1/2 \\
       (2) & (0) & (2) & 1 \\
       (2) & (1) & (3) & 1 \\
       (2) & (4,0,0) & (4,1,1) & 1/2 \\
       (2) & (4,0,1) & (4,1,0) & 1/2 \\
       (2) & (4,1,0) & (4,0,1) & 1/2 \\
       (2) & (4,1,1) & (4,0,0) & 1/2 \\
       (3) & (0) & (3) & 1 \\
       (3) & (1) & (2) & 1 \\
       (3) & (4,0,0) & (4,0,0) & 1/2 \\
       (3) & (4,0,1) & (4,0,1) & -1/2 \\
       (3) & (4,1,0) & (4,1,0) & -1/2 \\
       (3) & (4,1,1) & (4,1,1) & 1/2 \\
       (4,0,0) & (0) & (4,0,0) & 1 \\
       (4,0,1) & (0) & (4,0,1) & 1 \\
       (4,1,0) & (0) & (4,1,0) & 1 \\
       (4,1,1) & (0) & (4,1,1) & 1 \\
       (4,0,0) & (1) & (4,1,1) & 1 \\
       (4,0,1) & (1) & (4,1,0) & -1 \\
       (4,1,0) & (1) & (4,0,1) & -1 \\
       (4,1,1) & (1) & (4,0,0) & 1 \\
       (4,0,0) & (2) & (4,1,1) & 1 \\
       (4,0,1) & (2) & (4,1,0) & 1 \\
       (4,1,0) & (2) & (4,0,1) & 1 \\
       (4,1,1) & (2) & (4,0,0) & 1 \\
       (4,0,0) & (3) & (4,0,0) & 1 \\
       (4,0,1) & (3) & (4,0,1) & -1 \\
       (4,1,0) & (3) & (4,1,0) & -1 \\
       (4,1,1) & (3) & (4,1,1) & 1 \\
    \end{tabular}
    \caption{Clebsh-Gordan coefficients for $D_4$ group.}
    \label{tab:CG_D4}
\end{table}
Furthermore, the elements of the individual link variables in the irrep basis can
be explicitly written as
\begin{align}\label{eq:Uirreps}
    \hat{U}_{0,0} &= \frac{1}{\sqrt{2}}\Big[ \Big(\ket{\widetilde{0}}+\ket{\widetilde{3}}\Big)\bra{\widetilde{4,0,0}} + \Big(\ket{\widetilde{0}}+\ket{\widetilde{2}}\Big)\bra{\widetilde{4,1,1}}\Big]\\
    \hat{U}_{0,1} &= \frac{1}{\sqrt{2}}\Big[ \Big(\ket{\widetilde{0}}-\ket{\widetilde{3}}\Big)\bra{\widetilde{4,0,1}} + \Big(\ket{\widetilde{2}}-\ket{\widetilde{1}}\Big)\bra{\widetilde{4,1,0}}\Big]\\
    \hat{U}_{1,0} &= \frac{1}{\sqrt{2}}\Big[\Big(\ket{\widetilde{2}}-\ket{\widetilde{1}}\Big)\bra{\widetilde{4,0,1}} + \Big(\ket{\widetilde{0}}-\ket{\widetilde{1}}\Big)\bra{\widetilde{4,1,0}}\Big]\\
    \hat{U}_{1,1} &= \frac{1}{\sqrt{2}}\Big[ \Big(\ket{\widetilde{1}}+\ket{\widetilde{2}}\Big)\bra{\widetilde{4,0,0}} + \Big(\ket{\widetilde{0}}+\ket{\widetilde{3}}\Big)\bra{\widetilde{4,1,1}}\Big].
\end{align}
In general, Eq.~\eqref{eq:Uirreps} can used in combination with the Clebsch-Gordan coefficients to obtain explicit expressions for the plaquette terms
of the Hamiltonian in irrep basis.

\section{Density condition for the free products of random elements of a perfect group}\label{app:twogens}
The main goal of this Section is to show that the free product of two random elements of a perfect group is dense in the same group with probability 1.
For our purposes, we state the result in the case of the $SU(N)$ group, 
which maps to a sequence of quantum unitary gates whose concatenation (circuit) acts transitively on a $N$-dimensional subspace (i.e., the physical Hilbert space of the gauge theory).

\begin{theorem}\label{thm:twoDenseSUN}
    Let us consider two elements $R_1,R_2$ drawn 
    randomly and independently 
    from $G=SU(N)$ in such a way that the 
    probabilities are non-zero on the whole group (i.e.~the 
    distributions involved have compact support everywhere).
    Then the set generated by the free product
    of $R_1$ and $R_2$, 
    i.e. $\{{\prod_{i=1}^{k} (R_1^{p_i} 
    R_2^{q_i})}|p_i,q_i\in\mathbb{Z}, k\in \mathbb{N}\}$, 
    is dense in $SU(N)$.
\end{theorem}
\begin{proof}
In Theorem 6 of Ref.~\cite{Kuranishi:1951},
it is proved that for a semi-simple Lie algebra 
$\mathscr{L}$ there exists two elements $a,b\in \mathscr{L}$
generating the whole algebra. 
This connects to Theorems 7 and 8 of Ref.~\cite{Kuranishi:1951},
which states that if $G$ is a connected and 
perfect\footnote{
A perfect group is a group $G$ 
for which the commutator subgroup coincides with the group
itself, i.e.~$[G,G]=\{ghg^{-1}h^{-1}| g,h\in G\}=G$.} 
Lie group with lie algebra $\mathscr{L}$, 
and $\mathscr{L}$ is generated by two elements, 
therefore, also $G$ is generated by two elements, which
can be taken in an arbitrarily small neighborhood 
of the identity, and the subgroup generated by these 
is everywhere dense in $G$. 
Since $SU(N)$ is perfect and connected, 
it satisfies the assumptions of Theorems 6, 7 and 8 in Ref.~\cite{Kuranishi:1951}.
We are only left to prove that drawing any 
two random elements $R_1$ and $R_2$, 
they are finite generators of the whole group with probability $1$.
A sketch of the proof for this is 
inspired by the proof of Theorem 6 in 
Ref.~\cite{Kuranishi:1951}, which proceeds through 
an explicit construction of generators of the Lie 
algebra $\mathscr{L}$.
Let us consider a basis $\{h_i\}_{i=1}^{l}$ for the Cartan subalgebra $\mathscr{H}\subset \mathscr{L}$,
and elements $e_{\vec{\alpha}}$ associated 
to root vectors $\vec{\alpha}$ such that 
$\{h_i\}\cup \{e_{\vec{\alpha}}\}$ is a basis for $\mathscr{L}$ 
(see the pedagogical Ref.~\cite{georgi1982lie} for 
an introduction to Cartan subalgebras and root vectors).
Therefore, one can build two
Lie algebra generators $a\equiv \sum_{\vec{\alpha}} e_{\vec{\alpha}}$ 
and $h\equiv \sum_i \lambda_i h_i$ 
whose dynamical algebra built from successive
commutators
\begin{equation}
\begin{aligned}
    s_1&=[h,a]=\sum_{\vec{\alpha}} (\vec{\alpha} \cdot{\vec{\lambda}}) e_{\vec{\alpha}},\\
    s_2&=[h,[h,a]]=\sum_{\vec{\alpha}} {(\vec{\alpha} \cdot{\vec{\lambda}})}^2 e_{\vec{\alpha}},\\
    &\qquad\qquad\qquad \vdots\\
    s_k&=\underbrace{[h,\dots,[h}_{k\text{ times}},a]\dots]=\sum_{\vec{\alpha}} {(\vec{\alpha} \cdot{\vec{\lambda}})}^k e_{\vec{\alpha}},
\end{aligned}
\end{equation}
generates the whole Lie algebra 
(see Ref.~\cite{Kuranishi:1951} for details), 
provided the values $\lambda_i$
are chosen such that 
$(\vec{\alpha}-\vec{\beta})\cdot \vec{\lambda}\neq 0$ 
for every root $\vec{\alpha}$ and $\vec{\beta}$.
If the values $\lambda_i$ are drawn randomly and
independently, this requirement is satisfied 
with probability 1 and we are done.
Let us consider the two random special unitaries $R_1$
and $R_2$. Since $R_1$ is normal, it can be diagonalized
by some unitary $T$, as $\tilde{Z}\equiv T^\dagger R_1 T$,
which is generated by some element $h$ of the 
Cartan subalgebra $\mathscr{H}$. 
Performing the same transformation on the other
generator $\tilde{X}\equiv T^\dagger R_2 T$
yields a non-diagonal operator with probability 1;
this is generated by a sum of some element of the Cartan
subalgebra $\tilde{h}$, in irrational relation with
$h$ with probability 1, and with a linear combination
of all root elements, again with probability 1.
A free product between $R_1$ and $R_2$
is then equivalent to a free product of 
$\tilde{Z}$ and $\tilde{X}$ followed by a conjugation
by $T$. Since $\tilde{Z}$ and $\tilde{X}$ satisfy
the requirements of the Theorems above, and the group
is perfect, also $R_1$ and $R_2$ can be used as 
free group generators of the whole group.
\end{proof}

\section{Overview of Gaussian Kernel Distribution Estimation}\label{app:GKDE}
Let us consider a dataset $\mathcal{D}=\{x_i\}_{i=1}^N$ 
of independent variables extracted with a probability distribution $p^{(e)}(x)$,
and let us consider a Gaussian kernel
\begin{align}
   G_\sigma(x,y) \equiv \frac{1}{\sqrt{2\pi} \sigma} e^{-\frac{{(x-y)}^2}{2 \sigma^2}}.
\end{align}
We define the kernel estimate at a point $y$ as follows: 
\begin{align}\label{eq:KDEestimator}
    p_\sigma(y;\mathcal{D}) = \frac{1}{N} \sum_{i=1}^N G_\sigma(y,x_i).
\end{align}
As for standard histograms with fixed bin size, 
the function in Eq.~\eqref{eq:KDEestimator} can be considered as a $\sigma$-coarse-grained estimator for the exact probability distribution, 
using a Gaussian kernel, i.e.:
\begin{align}
    p_\sigma^{(e)}(y)&\equiv \expval{p_\sigma(y;\mathcal{D})}_{\mathcal{D}}\\ 
    &= \frac{1}{N} \sum_{i=1}^{N} \int\! dx_i \; p^{(e)}(x) G_\sigma(y,x) \underset{\substack{N\to \infty \\ \sigma \to 0^+}}{\to} p^{(e)}(y),
\end{align}
where ${\langle f(\{x_i\}) \rangle}_{\mathcal{D}}=\int \prod_i [d x_i p^{(e)}(x_i)] f(\{x_i\})$ is the expectation value with respect to all datasets $\mathcal{D}$ (with fixed number of elements $N$ implied). 
In particular, for Gaussian kernels, the smeared distribution $p_\sigma^{(e)}(y)$ is connected to the exact probability distribution as series in powers of $\sigma^2$ through a saddle point expansion~\cite{Daniels:1954}:
\begin{align}
    {\langle p_{\sigma}(y;\mathcal{D})\rangle}_{\mathcal{D}} = \sum\limits_{r=0}^{\infty} {\Big(\frac{\sigma^2}{4}\partial_y^2\Big)}^{r} p^{(e)}(y)\underset{\sigma\to 0^+}{\longrightarrow} p^{(e)}(y).
\end{align}

In order to compute the associated error we can compute the variance of the kernel estimate:
\begin{align}
    \sigma^2_{p_\sigma(y;\mathcal{D})} &=  \expval{{\Big(\frac{1}{N}\sum_{i=1}^{N} G_\sigma(y,x_i)\Big)}^2}_{\mathcal{D}}-{\big(p_\sigma^{(e)}(y)\big)}^2\\
    &=\frac{1}{N} \Big[ \frac{p_{\sigma/\sqrt{2};\mathcal{D}}^{(e)}(y)}{2\sqrt{\pi}\sigma} - {\big(p_\sigma^{(e)}(y)\big)}^2\Big].
\end{align}
Therefore, for independent data, the (unbiased) error to be associated with each bin bar $j$:
\begin{align}
    {\Delta p_\sigma(y;\mathcal{D})}&\simeq \sqrt{\frac{1}{N-1} \Big[ \frac{p_{\sigma/\sqrt{2}}(y;\mathcal{D})}{2\sqrt{\pi}\sigma} - {\big(p_\sigma(y;\mathcal{D})\big)}^2\Big]}.
\end{align}
However, when some autocorrelation time is present in the data, in practice one 
it is useful to perform a blocked partition of the dataset $\mathcal{D}$, 
followed by a certain number $K$ of resamples $\mathcal{D}_s$ 
(i.e., jackknife or bootstrap) such that the statistical error is estimated as 
\begin{align}\label{eq:deltaRES}
{{\lbrack\Delta p_\sigma(y;\mathcal{D})}\rbrack}_{\text{resamp.}}=\sqrt{\frac{1}{K} \sum_{s=1}^{K} {[p^s_\sigma(y;\mathcal{D}_s)-p_\sigma(y;\mathcal{D}_s)]}^2}.
\end{align}
In general, the smoothing parameter $\sigma$ (also called bandwidth in Statistics literature) should be chosen to satisfy an optimality criterion, such as the minimization of the expected mean integrated squared error (MISE)
\begin{align}
    \text{MISE}(\sigma)&={\Big\langle\int dy {\big[p_\sigma(y;\mathcal{D}) - p^{(e)}(y)\big]}^2\Big\rangle}_{\mathcal{D}}.
\end{align}
A $\sigma$ too large results in higher bias from the exact distribution $p^{(e)}(y)$ (undersampling), while a $\sigma$ too small is also not recommended, since it results in higher variance among results from different datasets (oversampling). Several techniques can be used to estimate and minimize this quantity from a dataset $\mathcal{D}$ (see Refs.~\cite{Bowman:1984,Sheather:1991}),
yielding a general bound for $\sigma$ with finite statistics $N$ of the type $\sigma \gtrsim C N^{-\frac{1}{5}}$, 
for some coefficient $C$ which has to be estimated from data.
However, for the purposes of this paper (data from QMS energy measurements is constrained to lie on the QPE grid points), 
it is sufficient to ensure that the bandwidth is of the order of the grid spacing. 
In our case, we set the smoothing parameter to be always of the order of the grid spacing for each number of qubit $q_e$ for the energy register, 
therefore scaling exponentially as $\sigma_{\text{KDE}}=\frac{\Delta E^{(\text{grid})}}{2^{q_e}-1}$, while our statistics is sufficient to satisfy the bound.

\bibliographystyle{apsrev4-1}
\bibliography{refs}

\end{document}